\documentclass[12pt]{article}
\usepackage{float}
\usepackage{amsmath} %
\usepackage[dvips]{color}
\usepackage{subfigure}
\usepackage{graphicx}
\usepackage{amsmath}
\usepackage{framed}
\usepackage{enumerate}
\usepackage{bbm}
\usepackage{comment}
\definecolor{dgreen}{rgb}{0,.8,.3}
\definecolor{lblue}{rgb}{.2,.3,.7}
\definecolor{blue}{rgb}{0,0,0.5}

\newtheorem{theorem}{Theorem}

\newtheorem{lemma}{Lemma}

\newenvironment{proof}{{\bf Proof:}}{\hfill$\bull$\medskip}

\newcommand{\bull}{\vrule height 1.8ex width 1.0ex depth 0ex}

\begin{document}
\begin{center}
{\bf\Large Sparse Portfolio Selection via \\ \vspace{0.1in} Quasi-Norm Regularization} \\
 \vspace{0.2in}
Caihua Chen\footnote{International Center of Management Science and Engineering, School of Management and Engineering, Nanjing Univeristy, China. This author is partially supported by the Natural Science Foundation
of Jiangsu Province BK20130550 and the Natural Science Foundation of China
NSFC-71271112. Email: chchen@nju.edu.cn.
},
Xindan Li\footnote{International Center of Management Science and Engineering, School of Management and Engineering, Nanjing Univeristy, China. This author is partially supported by the Natural Science Foundation of China NSFC-70932003. Email: xdli@nju.edu.cn.},
Caleb Tolman\footnote{Department of Management Science
and Engineering, School of Engineering, Stanford University, USA.
Email: calebj@stanford.edu. This author is
partially supported by AFOSR Grant FA9550-12-1-0396.},
Suyang Wang\footnote{International Center of Management Science and Engineering, School of Management and Engineering, Nanjing Univeristy, China; and Department of Management Science
and Engineering, School of Engineering, Stanford University, USA. This author is supported by CSC. Email: suyangw@stanford.edu.},
Yinyu Ye\footnote{Department of Management Science
and Engineering, School of Engineering, Stanford University, USA; and International Center of Management Science and Engineering, School of Management and Engineering, Nanjing University.
Email: yyye@stanford.edu. This author is
partially supported by AFOSR Grant FA9550-12-1-0396.} \\
December 16, 2013.  
\end{center}
\noindent {\bf Abstract}
In this paper, we propose $\ell_p$-norm regularized models to seek
near-optimal sparse portfolios. These sparse solutions reduce the complexity of portfolio implementation and management. Theoretical results are established to guarantee the sparsity of the second-order  KKT points of the $\ell_p$-norm regularized models. More interestingly, we present a theory that relates sparsity of the KKT points with Projected correlation and Projected Sharpe ratio. We also design an interior point algorithm to obtain an approximate second-order KKT solution of the $\ell_p$-norm models in polynomial time with a fixed error tolerance, and then test our $\ell_p$-norm modes on  S\&P 500 (2008-2012) data and international market data.\ The computational results illustrate that the $\ell_p$-norm regularized models can generate portfolios of any desired sparsity with portfolio variance and portfolio return comparable to those of the unregularized Markowitz model with cardinality constraint. Our analysis of a \emph{combined} model lead us to conclude that sparsity is not \emph{directly} related to overfitting at all.  Instead, we find that sparsity moderates overfitting only \emph{indirectly}.  A combined $\ell_1$-$\ell_p$ model shows that the proper choose of leverage,  which is the amount of additional buying-power generated by selling short can mitigate overfitting; A combined $\ell_2$-$\ell_p$ model is able to produce extremely high performing portfolios that exceeded the 1/N strategy and all $\ell_1$ and $\ell_2$ regularized portfolios. 
 \vspace{0.1in}


\vspace{0.2in}
\newpage
\section{Introduction}
\medskip
The origin of modern portfolio theory can be traced back to the early 1950's, beginning with Markowitz's work \cite{Markowitz:1952} on mean-variance
formulation. Given a basket of securities, the Markowitz model seeks to find
the optimal asset allocation of the portfolio by minimizing
the estimated variance with an expected return above a specified level.

Although the Markowitz mean-variance model captures the most two essential aspects in
portfolio management---risk and return,  it is not trivial to implement  the model directly
in the real world. One of the most critical challenge is the overfitting problem. Overfitting arises from the inability to perfectly estimate the mean and covariance of real-world objects. In fact, due to high dimensionality and non-normal distribution of the unknown variable, these estimates are especially inaccurate for stock data. Indeed,  \cite{Merton:1980} shows that most of the difficulty lies on the mean estimate.  Moreover, \cite{DeMiguel:2009}
show that in order to estimate the expected return of  portfolio of 25 stocks with satisfactorily low error, one would need on the order of 3000 months of data, which is both extremely difficult to acquire and too long for the model to obey the time-invariance assumptions. The Markowitz model does nothing to prevent the overfitting that comes from mis-estimation, and thus performs poorly across most out-of-sample metrics. For example, \cite{Demiguel:2009a} evaluate the out-of-sample performance of the mean-variance model and find that none of algorithms to compute the solution of the Markowitz model consistently outperforms the naive ${1/N}$ (equal amounts of every stock)
portfolio.

To alleviate the overfitting, several variants of the Markowitz model with regularizers/additional constraints have been proposed in the literature. The modifications can be viewed as adding a prior belief on the true yet unknown return distributions (as suggested by \cite{Merton:1980}). In \cite{Jagannathan:2003}, the authors impose a non-shortsale constraint to the mean-variance formulation despite the fact that leading theory speaks against this constraint. Surprisingly, the ``wrong'' constraint helps the model to find solution with better out-of-sample performance. More recently, \cite{Brodie:2009} and \cite{Rosenbaum:2010} succeed in applying the $\ell_1$-norm technique to the Markowitz model to obtain sparse portfolios with higher Sharpe ratio and stability than the naive $1/N$ rule. By adding a norm ball constraint to the portfolio-weight vector, \cite{DeMiguel:2009} provide a general framework for determining the optimal portfolio. The computational results demonstrate that the norm ball constrained portfolios typically achieve lower out-of-sample variance and higher out-of-sample Sharpe ratio than the proposed strategies in \cite{Jagannathan:2003}, the naive $1/N$ portfolio and many others in the literature.

Meanwhile, the optimal portfolio of Markowitz's classical model often holds a huge number of assets and some assets admit extremely small weights. Such a solution, however, is not attainable in most situations of the real market. Due to physical, political and economical constraints,  investors  would be willing to sacrifice a small degree of performance for a more manageable sparse portfolio (see \cite{Shefrin:2000, Boyle:2012,Guidolin:2013} and references therein).  An illustrative example comes from the most successful investor of the 20th century,  Warren Buffet, who advocates investing in a few familiar stocks, which is also supported by the early work of Keynes (see \cite{Moggridge:1983}).

A popular way to construct the sparse portfolio is via the cardinality constrained portfolio selection (CCPS) model (\cite{Bertsimas:2009,Cesarone:2009,Maringer:2003}) , i.e., choose a specified number of assets to form an efficient portfolio. Unfortunately, the inherent combinatorial property makes the cardinality constrained problem NP-hard generally and hence computationally intractable. By relaxing the hard cardinality constraint, many heuristic methods \cite{Bienstock:1996,Chang:2000}  have been proposed to solve the CCPS. Very recently, by  relaxing the objective function as some separable functions, \cite{Gao:2013} obtain a cardinality constrained relaxation of CCPS with closed-form solution. The new relaxation combined with a  branch-and-bound algorithm (Bnb) yields a highly efficient solver, which outperforms CPLEX significantly.

The main objective of our paper is to propose a novel and non-CCPS portfolio strategy with complete flexibility in choosing sparsity while still maintaining satisfactory out-of-sample performance. Here, we discuss a new regularization of Markowitz's  portfolio construction both with and without the shortsale constraint. To accomplish this objective, we turn to the $\ell_p$-norm ($0<p<1$) regularization which recently attracts a growing interest from  the optimization community due to its important role in inducing sparsity. Theoretical and empirical results indicate that the $\ell_p$-norm regularization (\cite{Chartrand:2007,Xu:2009,Ji:2013,Saab:2008}) could have better stability and sparsity than the traditional $\ell_1$-norm regularization. In this work, we take a step to study the theoretical and computational performance of the $\ell_p$-norm regularized portfolio optimization problem in the framework of the Markowitz model.

The contributions of our paper include (i) a novel portfolio strategy to produce 50\%--95\% more sparse portfolios with competitive out-of-sample performance compared with the Markowitz model and the $\ell_1$-norm model; (ii) a polynomial time interior point algorithm to compute the second-order KKT solutions of our $\ell_p$-norm models; iii) an extension of the modern portfolio theory that relates sparsity to  ``Projected correlation'' and ``Projected Sharpe ratio''; (iv) an ``efficient frontier'' outlining the optimal tradeoff between sparsity and expected return and variance.

The remainder of this paper is organized as follows. In Section 2, we review some relevant portfolio models in the literature and present our $\ell_p$-norm regularized formulations for sparse portfolio selection  with/without shorting constraints.  In Section \ref{sec:theory}, we develop the $\ell_p$-norm regularization portfolio theory with financial interpretation, and design a fast interior point algorithm to compute the KKT points of our regularized models in polynomial time. We also construct some toy examples to show the intuition of our portfolio theory. Section \ref{sec:computationalresults} is devoted to the computational results of the regularized models and comparison between different models, which show our portfolio strategies have high sparsity but still maintain out-of-sample performance. Section \ref{sec:discussionsandconclusions} concludes our work and provides a possible application of our research.
All proofs of the propositions can be found in the Appendix I and the details of our interior point algorithm are described in the Appendix II.

\section{The Related Models}
\label{sec:therelatedmodels}
Given a  portfolio consisting of $n$  stocks.  The Markowitz mean-variance portfolio
is the solution of the following constrained optimization problem
\begin{equation}\label{MVI}
\begin{array}{rl}
\min &\displaystyle \, \frac{1}{2}\,x^TQx\\[0.3cm]
\mbox{s.t.}       &  e^Tx =1,\\[0.1cm]
& m^Tx  \ge m_0,
\end{array}
\end{equation}
where $Q\in \Re^{n\times n}$ is the estimated covariance matrix of the portfolio, $m\in \Re^n$ is
the estimated return vector,  $m_0\in \Re$ is a specific return level, and $e$ is the vector of all ones with a matching dimension. Note also that, if the non-shortsale constraint $x \geq 0$ is added to \eqref{MVI}, the resulting model is the formulation of the shorting-prohibited Markowitz model.  Assume the optimal Lagrangian multiplier associated with the mean constraint is known as $\phi$. Then we can recast the Markowitz model without (with) no-shorting constraint as a linear equality constrained optimization problem
\begin{equation}\label{MVI-1}
\begin{array}{rl}
\min &\displaystyle \, \frac{1}{2}\,x^TQx - c^T x \\[0.3cm]
\mbox{s.t.}       &  e^Tx =1,\\[0.1cm]
                        & ( x \ge 0 ),
\end{array}
\end{equation}
where $c = \phi m$.

 \cite{Brodie:2009}  discuss the $\ell_1$-norm regularized Markowitz model
\begin{equation}\label{MV-l1}
\begin{array}{rl}
\min &\displaystyle \, \frac{1}{2}\,x^TQx + \rho \|x\|_1\\[0.3cm]
\mbox{s.t.}       &  e^Tx =1,\\[0.1cm]
& m^Tx  = m_0.
\end{array}
\end{equation}
Here the $\ell_1$-norm of a vector $x\in \Re^n$ is defined by $\|x\|_1 :=\sum_{i=1}^n |x_i|$
and $\rho$ is a positive penalty parameter. Sparse portfolios can be obtained by solving \eqref{MV-l1} with increasing values of $\rho$. The $\ell_1$-norm, however, cannot be effective in conjunction with the no-shorting constraint, and thus it cannot induce sparsity beyond the sparsity of the no-shorting Markowitz portfolio. This fact can be explained as follows: let $x^+$ and $-x^{-}$ denote the positive and negative entries of $x$, respectively. Then, in order to satisfy the budget constraint, we must have:
\[e^Tx^+=e^Tx^{-}+1.\]
Since $\|x\|_1 = e^Tx^+ + e^Tx^-$, we also have that $\|x\|_1=2e^Tx^{-}+1$. Thus, adding $\|x\|_1$ into the objective penalizes shorting activity the sum of the absolute negative entries in $x$ and  thus has less effect as a penalty on sparsity.

Such a gap motivates us to study the following  concave $\ell_p$-norm $\,(0<p<1)$ regularization of the no-shorting mean-variance model
\begin{equation}\label{MV-lp}
 \begin{array}{rl}
\min &\displaystyle  \frac{1}{2}\,\,x^TQx - c^Tx + \lambda\|x\|^p_p \\[0.3cm]
\mbox{s.t.}       &                  e^Tx            =1,\\[0.1cm]
       &                       x             \ge 0,
\end{array}
\end{equation}
where the $\ell_p$-norm of $x\in \Re^n$ is defined as $\|x\|_p =\sqrt[p]{\sum_{j=1}^n |x_j|^p}$. And then when $x\geq 0$, $\|x\|_p^p= \sum_{j=1}^n x_j^p$.
It is noteworthy that  the $\ell_p$-norm regularized problem \eqref{MV-lp}
can be regarded as a continuous iterative heuristic of the following CCPS problem
\begin{equation}\label{MV-cardi}
 \begin{array}{rl}
\min & \frac{1}{2}x^TQx-c^Tx\\[0.2cm]
\mbox{s.t.}       &  e^Tx =1,\\[0.1cm]
&\|x\|_0\leq K,\\[0.1cm]
& x\geq 0,
\end{array}
\end{equation}
where $\|x\|_0$ represents the number of the nonzero entries of $x$ and $K$ is the chosen limit of stocks to be managed in the portfolio.

We also study the portfolio selection problem with the no-shorting constraint removed.  Analogues to the above models, we consider the following $\ell_p$-norm model
\begin{equation}\label{MV-lp-short}
 \begin{array}{rl}
\min &\displaystyle  \frac{1}{2}\,\,x^TQx - c^Tx + \lambda\|x\|^p_p \\[0.3cm]
\mbox{s.t.}       &                  e^Tx            =1.
\end{array}
\end{equation}
Moreover, \cite{DeMiguel:2009} construct
the optimal portfolio with high Sharpe ratio via solving the following the
minimum-variance problem subject to a norm ball constraint, i.e.,
\begin{equation}\label{MV-l1con}
\begin{array}{rl}
\min &\displaystyle \, \frac{1}{2}\,x^TQx \\[0.3cm]
\mbox{s.t.}       &  e^Tx =1,\\[0.1cm]
                        &  \|x\|\leq \delta,
\end{array}
\end{equation}
where $\delta$ is a given threshold. Following this work and specifying the general norm as the $\ell_1$-norm, we propose the $\ell_1$-norm ball constrained  the $\ell_p$-norm regularized
Markowitz model
\begin{equation}\label{MV-lpreg-l1}
\begin{array}{rl}
\min & \frac{1}{2}x^TQx - c^Tx + \lambda\|x\|_p^p\\[0.2cm]
\mbox{s.t.} & e^Tx =1,\\[0.1cm]
& \|x\|_1 \le \delta.
\end{array}
\end{equation}
By splitting the vector $x:=x^+-x^{-}$,   \eqref{MV-lpreg-l1} can be equivalently written as
\begin{equation}\label{MV-lpreg-l1con}
\begin{array}{rl}
\min & \frac{1}{2}(x^+-x^{-})^TQ(x^+ - x^{-}) - c^T(x^+-x^{-}) + \lambda\|x^+\|_p^p+ \lambda \|x^{-}\|^p_p\\[0.2cm]
\mbox{s.t.} & e^Tx^+-e^T x^- =1,\\[0.1cm]
& e^T x^+ + e^T x^- \le \delta,\\[0.1cm]
& x^+ \ge 0,\, x^- \ge 0.
\end{array}
\end{equation}
Besides, we also consider the following $\ell_2-\ell_p$-norm double regularization Markowitz model
which can be seen as a Lagrangian form of \eqref{MV-l1con} with a $\ell_2$-norm ball
\begin{equation}\label{MV-lp-l2}
\begin{array}{rl}
\min & \frac{1}{2}x^TQx - c^Tx + \lambda\|x\|_p^ p + \mu \|x\|_2^2 \\[0.2cm]
\mbox{s.t.} & e^Tx =1,
\end{array}
\end{equation}
as well as its splitting form
\begin{equation}\label{MV-lp-l2-split}
\begin{array}{rl}
\min & \frac{1}{2}(x^+-x^{-})^TQ(x^+ - x^{-}) - c^T(x^+-x^{-})  +\lambda\|x^+\|_p^p +\lambda\|x^-\|_p^p + \mu \|x^+ - x^-\|_2^2 \\[0.1cm]
\mbox{s.t.} & e^Tx^+ -  e^T x^- =1,\\[0.1cm]
& x^+ \ge 0,\, x^- \ge 0,
\end{array}
\end{equation}
where $x = x^+-x^-$.  It can  be shown later that the regularization models \eqref{MV-lpreg-l1con}  and \eqref{MV-lp-l2-split}  always produces a complementary pair $x^+$ and $x^{-}$: that is, $x^+_jx^{-}_j=0$ for all $j$. In this paper, we develop theories on the models as well as computation evidences that the models produce sparse portfolios with high out-of-sample Sharpe ratio.

\section{$\ell_p$-norm Regularized Portfolio Theory} \label{sec:theory}
\setcounter{equation}{0}
\medskip
In this section, we develop theoretical results on the sparsity of the  $\ell_p$-norm regularized
models with toy examples to illustrate the intution and also provide  financial interpretation of the theory. 
Our approach to establish the 
theoretical results is motivated by the results (\cite{Chen:2010}) in singal processing.
For simplicity, hereafter we will fix $p= 1/2$.

\subsection{Bounds of Nonzero Elements of KKT Points}
First, we develop bounds on the non-zero entries of any KKT solution of the $\ell_p$-norm regularized
Markowitz model with the  non-shortsale constraint $x\geq 0$.
\begin{theorem} \label{secondordertheorem}
Let $\bar{x}$ be any second-order KKT solution of \eqref{MV-lp},
that is, a first-order KKT solution that also satisfies the second-order necessary condition, $\bar{P}$ be the support of $\bar{x}$
and $\bar{Q}$ be the corresponding covariance sub-matrix. Furthermore, let $K=|\bar{P}|$ and
\[L_i=\bar{Q}_{ii} -\frac{2}{K}(\bar{Q}e)_i+\frac{1}{K^2}(e^T\bar{Q}e),\ i\in \bar{P},\]
which are the diagonal entries of the projection of $\bar{Q}$ onto the null space of vector $e$:
\[\left(I-\frac{1}{K}ee^T\right)\bar{Q}\left(I-\frac{1}{K}ee^T\right).\]
Then it holds that
\begin{description}
\item {\rm (i)}
\[
(K-1)K^{3/2}\le \frac{4\sum_{i\in \bar{P}}L_i}{\lambda}
=\frac{4}{\lambda}\Big[{\rm tr}(\bar{Q})-\frac{1}{K}e^T\bar{Q}e\Big]
.\]
\item{\rm (ii)} If $L_i=0$ for some $i\in \bar{P}$, then $K=1$ so that $\bar{x}_i=1$; otherwise,
\[
 \bar{x}_i\geq \left (\frac{\lambda(K-1)^2}{4L_iK^2}\right)^{2/3}
\]
\end{description}
\end{theorem}
\begin{proof} Please see the proof in the Appendix I. 
\end{proof}

Note that if $\sum_{i\in \bar{P}}L_i=0$, the first statement of our theorem implies that $K=1$.
This can be explained as follows. $\sum_{i\in \bar{P}}L_i=0$ implies the projected
$\bar{Q}$ matrix
\[\left(I-\frac{1}{K}ee^T\right)\bar{Q}\left(I-\frac{1}{K}ee^T\right)=0.\]
Then, $\bar{Q}=\alpha ee^T$ for some $\alpha\ge 0$,
in which case the portfolio variance $\bar{x}^T\bar{Q}\bar{x}=\alpha$ and it is a constant. Thus, the optimal solution of the regularized problem would allocate $100\%$ into the stock with the highest $c_i$ or highest return factor.  Our theorem also implies that the greater of $\lambda$, the less of $K$. The quantity of $\sum_{i\in \bar{P}}L_i$ represents the total diversification coefficient of the set of stocks $i\in \bar{P}$; the smaller of the quantity, the less the size of $\bar{P}$ -- the set of selected stocks in the portfolio by the $\ell_p$ norm regularized
Markowitz model.

The second statement provides an even stronger notion: if any $L_i=0,\,i\in \bar{P}$, then $K=1$. Basically, it says that investing only into
the $i$th stock suffices, since no diversification can help in this case.  Note that $L_i$ can be interpreted as other stocks' correlation to
stock $i$. If $L_i=0$, then other stocks present no diversification to the $i$th stock.

Next, we move to the  $\ell_1$-norm ball constrained $\ell_p$-norm regularized  Markowitz model and the double regularized model. The following theorems characterize the bound of nonzero elements of any second-order KKT points of problem \eqref{MV-lpreg-l1con} and \eqref{MV-lp-l2}.
\begin{theorem} \label{secondordertheorem1} Let $\bar{x} =(\bar{x}^+,\bar{x}^{-})$ be any second-order KKT solution of problem
\eqref{MV-lpreg-l1con} with $\delta>1$, $\bar{P}^+$ and  $\bar{P}^-$ be the support of $\bar{x}^+$ and $\bar{x}^-$,
and $\bar{Q}^+$ and  $\bar{Q}^-$ be the corresponding covariance sub-matrices, respectively. Furthermore, let
$K^+=|\bar{P}^+|$ and $K^-=|\bar{P}^-|$,  and
\[
L^j_i=\bar{Q}^j_{ii} -\frac{2}{K^j}(\bar{Q}^je)_i+\frac{1}{(K^j)^2}(e^T\bar{Q}^je),\ i\in \bar{P}^j,\,\, for\,\,  j\in \{+,-\},
\]
which are the diagonal entries of the projection of $\bar{Q}^j$ onto the null space of vector $e$:
\[
\left(I-\frac{1}{K^j}ee^T\right)\bar{Q}^j\left(I-\frac{1}{K^j}ee^T\right).
\]
Then it holds that
\begin{description}
\item {\rm (i)}
\[\bar{P}^+\cap \bar{P}^-=\emptyset.\]

\item {\rm (ii)}
\[
(K^+ -1)(K^+)^{3/2}\le \left({\delta+1\over 2}\right)^{3/2}\frac{4\sum_{i\in \bar{P}^+}L_i}{\lambda}
=\left({\delta+1\over 2}\right)^{3/2}\frac{4}{\lambda}\Big[{\rm tr}(\bar{Q}^+)-\frac{1}{K^+}e^T\bar{Q}^+e\Big] \]
and
\[
(K^- -1)(K^-)^{3/2}\le \left({\delta-1\over 2}\right)^{3/2}\frac{4\sum_{i\in \bar{P}^-}L_i}{\lambda}
=\left({\delta-1\over 2}\right)^{3/2}\frac{4}{\lambda}\Big[{\rm tr}(\bar{Q}^-)-\frac{1}{K^-}e^T\bar{Q}^-e\Big]. \]

\item{\rm (iii)}  If $L_i=0$ for some $i\in \bar{P}^+$ (or $i\in \bar{P}^-$), then $K^+=1$ (or $K^-=1$);  otherwise,
\[
 \bar{x}^j_i\geq \left (\frac{\lambda(K^j-1)^2}{4L^j_i(K^j)^2}\right)^{2/3}, \, \, i\in \bar{P}^j,\,\,{\rm for}\,\,j\in\{+,-\}.
\]
\end{description}
\end{theorem}
\begin{proof} Please see the proof in the Appendix I. 
\end{proof}
\medskip

\begin{theorem} \label{secondordertheorem2} Let $\bar{x} =(\bar{x}^+,\bar{x}^{-})$ be any second-order KKT solution of \eqref{MV-lp-l2-split},  $\bar{P}^+$ and  $\bar{P}^-$ be the support of $\bar{x}^+$ and $\bar{x}^-$,
 and $\bar{P} = \bar{P}^+\cup \bar{P}^-$. Furthermore, let $\bar{Q}$  be the covariance sub-matrices corresponding to $\bar{P}$, $K = |\bar{P}|$, and
\[
L_i=\bar{Q}_{ii} +2\mu-\frac{2}{K}(\bar{Q}e)_i+\frac{1}{(K)^2}(e^T\bar{Q}e),\ i\in \bar{P}.
\]
Then it holds that
\begin{description}
\item {\rm (i)}
\[\bar{P}^+\cap \bar{P}^-=\emptyset.\]
\item{\rm(ii)} If $\|\bar{x}\|_2\leq \delta$, then
\[
(K-1)K^{3/4}\le \frac{4 \delta^{3/2} \sum_{i\in \bar{P}}L_i}{\lambda}
=\frac{4\delta^{3/2}}{\lambda}\Big[{\rm tr}(\bar{Q})-\frac{1}{K}e^T\bar{Q}e\Big]
.\]


\item{\rm (ii)}  If $L_i=0$ for some $i\in \bar{P}$, then $K=1$
so that $\bar{x}_i= 1$ and $i\in \bar{P}^+$; otherwise,
\[
 \bar{x}^j_i \geq \left (\frac{\lambda(K-1)^2}{4L_i K^2}\right)^{2/3}, \, \, i\in \bar{P}.
\]
\end{description}
\end{theorem}
\begin{proof} Please see the proof in the Appendix I. 
\end{proof}
\medskip

The theories developed above indicate the importance to compute a second-order KKT solution, rather than just a first-order KKT solution, of the $\ell_p$-norm regularized portfolio management problems \eqref{MV-lp} and \eqref{MV-lpreg-l1con}. In this paper, we  present an interior point algorithm to compute an approximate second KKT point in polynomial time with a fixed error tolerance; see details in the Appendix II. The overall idea of using the interior-point algorithm is to start from a fully supported portfolio $x$ (that is, $x>0$) of every stock in consideration and iteratively eliminate a fraction of stocks at the end of the process.

\subsection{Characteristics of $L_i$}
In the theory supporting our model (see Section 3.1), there arose several interesting facts and characteristics to note about the
``Projected variances'' --- $\{L_i\}$ over the support set of a portfolio selected by the $\ell_p$-norm regularized
Markowitz models.

Given any stock portfolio, with the non-zero portion denoted as $x$, having support $P$ of size $K$ one can rewrite the quantity $L_i$ in Theorem \ref{secondordertheorem}, as follows:
\begin{equation}
\label{eq:interpretationofLi}
L_i = (e^i - e^0)^T\bar{Q}(e^i-e^0) = {\rm Var}\, [\eta^T(e^i-e^0)],
\end{equation}
$e_i\in R^K$ is the vector of all zeros except $1$ at the $i$th position and $e^0=\frac{1}{K}e\in R^k$.
Here $e^i$  and $e^0$ are the respective distributions obtained by investing 100\% in stock $i$ and
$\frac{1}{K}$ in each stock of the portfolio $x$, and $\eta$ represents the  random return vector of the portfolios.
Note that $L_i$, $i=1,...,K$, is independent of the entry values of $x$.

The difference vector
$(e^i-e^0)$ can be viewed as the ``cost-neutral portfolio action'' that sells an equal amount of everything in the current portfolio and uses all those funds to buy exactly one stock, stock $i$, within the current portfolio. Thus,
$L_i$ estimates the variance of this action. Let us now consider the feasible and optimal solutions of the Markowitz Model in Lagrangian form:
\begin{equation}\label{eq:genericportfolio}
\begin{array}{rl}
\min & \frac{1}{2}x^TQx-\phi m^Tx\\[0.2cm]
\mbox{s.t.}       &  e^Tx =1,\\[0.1cm]
& x  \ge 0,
\end{array}
\end{equation}
where $\phi$ is Lagrangian multiplier associated with the expected return inequality.

For any distribution portfolio---the non-zero portion denoted as $x$---one can plot the objective function of moving in a feasible exchange direction $e^i-e^0$:
\begin{align}
\label{eq:quadraticvariance}
f[x+\varepsilon(e^i-e^0)] &=\frac{1}{2} [x+\varepsilon(e^i-e^0)]^TQ[x+\varepsilon(e^i-e^0)] - \phi m^T[x+\varepsilon(e^i-e^0)] \\[0.2cm] \nonumber
&=f(x) + \varepsilon\, {\rm Cov} \,[x^T\eta,(e^i-e^0)^T\eta]  - \varepsilon\phi(\bar{m}_i-\bar{m}_0) +  \frac{1}{2}\varepsilon^2 L_i \nonumber
\end{align}

We now consider which stock would increase the variance the least when we remove it from that portfolio $x$.  Suppose we remove stock $i$ in the direction $e^i-e^0$, then we have a new portfolio support $P/\{i\}$ with distribution $x' = x-\frac{Kx_i}{K-1}(e^i-e^0)$.
Equation \eqref{eq:quadraticvariance} would give us the
\begin{equation}
\label{eq:marginalcostsofsparsity}
\text{Marginal Costs of Sparsity}
\end{equation}
\[ MCS_i = -\frac{K}{K-1}x_i[{\rm Cov}(x^T\eta,(e^i-e^0)^T\eta) - \phi(\bar{m}_i-\bar{m}_0)] + (\frac{K}{K-1})^2x_i^2L_i.\]
These marginal costs are only \emph{upper-bounds} on the true costs of sparsity. They do not consider any further improvement that could be made by re-balancing, and thus over-estimate costs.

When our current portfolio $x$ is a near-KKT point or local minimizer, we know from the first-order conditions that the first part
$[{\rm Cov}\,(x^T\eta,(e^i-e^0)^T\eta) - \phi(\bar{m}_i-\bar{m}_0)]$ must be near zero and thus the second order term will be a good approximation for the Marginal Cost by itself.  Hence, at a near (locally) optimal portfolio $x$, the best candidate for removal can be found by searching for the smallest values of $x_i\sqrt{L_i}$.
\begin{equation}
\text{Relative Sparsity Cost Index}
\end{equation}
\[{\rm  RSC}_i = x_i\sqrt{L_i} \]
Where the smallest non-zero RSC index is the cheapest (on the margin) to eliminate from $x$, and is likely to be the cheapest (absolutely) to remove.
Thus the quantity $L_i$ can be viewed as measures of elasticity: they indicate how sensitive the objective value is to small cost-neutral changes in $x$; small $L_i$ values therefore indicate which stocks could be removed from the portfolio with lowest cost.

\subsection{The Financial Interpretation of $f'(x;e^i-e^0)$ and $\varepsilon_i^*$}
\label{sec:projectedgradientandepsilon}

\quad
The cost-neutral portfolio actions $\big \{ e^i-e^0 \,| \,i \in [1,n] \big \}$ form a basis of the feasible directions, and thus the directional derivatives of the objective along these directions form a method of sensitivity analysis.  
\begin{equation}
\label{eq:directionalderivative}
f'(x,e^i-e^0) = {\rm Cov}[x^T\eta,(e^i-e^0)^T\eta] - \phi(m_i-\bar{m})
\end{equation}
Where $\bar{m}$ is the average of the expected returns of all stocks in the \emph{support} of $x$.  At an optimal point, these derivatives must be zero.  And for small deviations from optimality, these values can be used to approximate any smooth continuous function of the optimal solution.

Next, let's pay more attentions to the optimal step-size along the basic feasible directions. Specifically, given the direction $e^i -e^0$, the corresponding optimal stepsize $\varepsilon_i^*$ is given by
\begin{equation}\label{eq:optisetp}
\varepsilon_i^* =  -\frac{f'(x;e^i-e^0)}{L_i},
\end{equation}
which follows directly from \eqref{eq:quadraticvariance}.
At any optimal point the directional derivative is zero and thus the optimal step-size is zero; but if we were to consider a small change in the projected gradient,  $\varepsilon_i^*$ estimates the changes in optimal solution by taking the direction $e^i-e^0$.

 By substituting \eqref{eq:interpretationofLi} and \eqref{eq:directionalderivative} into \eqref{eq:optisetp}, we obtain
\[
 \varepsilon_i^*=\phi\frac{m_i-\bar{m}_*}{{\rm Var}\,[(e^i-e^0)^T\eta]} - \frac{{\rm Cov}\,[x^T\eta,(e^i-e^0)^T\eta]}{{\rm Var}\,[(e^i-e^0)^T\eta]}.
\]
The two parts can be easily related to the concepts  ``Projected correlation" and ``Projected Sharpe ratio", where the Projected correlation is
\begin{equation}
 \bar{\rho}_i := \frac{{\rm Cov}\,[x^T\eta,(e^i-e^0)^T\eta]}{ {\rm Std}\,[x^T\eta]* {\rm Std}\,[(e^i-e^0)^T\eta]},
\end{equation}
and
the Projected Sharpe ratio is
\begin{equation}
\bar{S}_i = \frac{m_i-\bar{m}_*} { {\rm Std}\,[(e^i-e^0)^T\eta]} .
\end{equation}
Then the optimal step size can be equivalently written as:
\begin{equation}
\varepsilon_i^* = \bar{S}_i\frac{\phi}{ {\rm Std}\,[(e^i-e^0)^T\eta]}
-\bar{\rho_i}\frac{ {\rm Std}\,[x^T\eta]} { {\rm Std}\,[(e^i-e^0)^T\eta]}
\end{equation}
It is clear that the optimal step size is sensitive to the inverse of the standard-deviation of the cost-neutral portfolio (inversely), as well as to the current portfolio standard-deviation.  The Projected correlation and Projected Sharpe ratio (as well as $\phi$) give the exact coefficients of these relationships.
\subsubsection{Toy Examples}
In this section, we illustrate the previous sensitivity analysis  by some dummy examples. Consider
the first example in Table \ref{tab:toy1}, where the portfolios include three stocks with identically distributed
variance yet differing expected returns. The lower returning stock admits a slightly smaller percentage (32.33\% vs 34.33\%) in the optimal portfolio due to the small reward ($\phi =0.01$) for the expected return. Since the RSC of stock 1 attains the minimum  cost of the three stocks,
according to our sensitive analysis, the investor would intuitively decrease the investment in the first stock further
and  thus remove the first stock from the basis to form a sparse portfolio (with the increasing of $\lambda$ ). Direct calculation also shows that this is the lowest cost stock to remove.
\begin{table}[h]
\centering
\caption{Toy example 1}
\label{tab:toy1}
\begin{tabular}{c c c c c c}
Mean & Variance & $x^* (\phi=0.01)$ & $L_i$ & OK to drop & RSC  \\
$ \begin{bmatrix}
1 \\ 2 \\ 3 \\
\end{bmatrix} $
&
$ \begin{bmatrix}
2 & 1 & 1 \\ 1 & 2 & 1 \\ 1 & 1 & 2 \\
\end{bmatrix} $
&
$\begin{bmatrix}
0.3233 \\ 0.3333 \\ 0.3433 \\
\end{bmatrix}$
&
$\begin{bmatrix}
0.6667 \\ 0.6667 \\ 0.6667 \\
\end{bmatrix}$
&
$\begin{bmatrix}
Yes \\ No \\ No \\
\end{bmatrix}$
&
$\begin{bmatrix}
0.2640 \\ 0.2722 \\ 0.2803 \\
\end{bmatrix}$
\end{tabular}
\end{table}

Consider the portfolio in Table \ref{tab:toy2}, where two stocks are positively correlated yet a third stock is independent;
all the stocks share a common mean and variance. The large value of  $L_1$   (see the MCS equation in \eqref{eq:marginalcostsofsparsity}) suggests that the first stock may not be a good candidate  to be removed, which can seen clearly by comparing the variances of the portfolios with  two stocks.

\begin{table}[h]
\centering
\caption{Toy example 2}
\label{tab:toy2}
\begin{tabular}{c c c c c c c}
Mean & Variance & $x^* (\phi=0.01)$ & $L_i$ & OK to drop & RSC\\
$ \begin{bmatrix}
0 \\ 0 \\ 0 \\
\end{bmatrix} $
&
$ \begin{bmatrix}
2 & 0 & 0 \\ 0 & 2 & 1 \\ 0 & 1 & 2 \\
\end{bmatrix} $
&
$\begin{bmatrix}
0.4355 \\ 0.2823 \\ 0.2823 \\
\end{bmatrix}$
&
$\begin{bmatrix}
1.5556 \\ 0.8889 \\ 0.8889 \\
\end{bmatrix}$
&
$\begin{bmatrix}
No \\ Yes \\ Yes \\
\end{bmatrix}$
&
$\begin{bmatrix}
0.5431 \\ 0.2661 \\ 0.2661 \\
\end{bmatrix}$
\end{tabular}
\end{table}

Table \ref{tab:toy3} lists the portfolio consisting of three stocks, where the third stock is actually a zero-cost mutual fund---one that simply invest equally in the first and second stocks. This third stock creates redundancy and thus \emph{infinitely} many optimal solutions are possible (we have shown one arbitrarily).  If we were to drop either the second stock or the third (but not both) from the portfolio, then we would still be able to attain the same optimal objective (75\%-25\% mix of Stock 1 and Stock 2 respectively for this small $\phi$, and a more balanced mix  larger $\phi$).  Moreover, we see
 that $L_3 = 0$, and this fact correctly predicts that there exists a strictly sparser optimal portfolio.
\begin{table}[h]
\centering
\caption{Toy example 3}
\label{tab:toy3}
\begin{tabular}{c c c c c c c}
Mean & Variance & $x^* (\phi=0.01)$ & $L_i$ & OK to drop & RSC\\
$ \begin{bmatrix}
1 \\ 3 \\ 2 \\
\end{bmatrix} $
&
$ \begin{bmatrix}
3 & 1 & 2 \\ 1 & 7 & 4 \\ 2 & 4 & 3 \\
\end{bmatrix} $
&
$\begin{bmatrix}
0.6875 \\ 0.1925 \\ 0.1200 \\
\end{bmatrix}$
&
$\begin{bmatrix}
2 \\ 2 \\ 0 \\
\end{bmatrix}$
&
$\begin{bmatrix}
No \\ BEST \\ OK \\
\end{bmatrix}$
&
$\begin{bmatrix}
0.9722 \\ 0.2722 \\ 0.00 \\
\end{bmatrix}$
\end{tabular}
\end{table}

As a last example, consider Table \ref{tab:toy4}, where we have a set of stocks that include two of them with high variance and positive correlation to most other stocks, yet highly negative correlation with each other.  These two stocks alone would make an excellent portfolio of size two.
\begin{table}[t]
\centering
\caption{Toy example 4}
\label{tab:toy4}
\begin{tabular}{c c c c c c}
mean & Variance & $x^* (\phi=0.01)$ & $L_i$ & OK to drop & RSC \\
$ \begin{bmatrix}
0 \\ 0 \\ 0 \\ 0\\
\end{bmatrix} $
&
$ \begin{bmatrix}
8 & 7 & 6 & 6 \\ 7 & 26 & 6 & 0 \\ 6 & 6 & 96 & -68 \\ 6 & 0 & -68 & 73 \\
\end{bmatrix} $
&
$\begin{bmatrix}
0.2913 \\ 0.1166 \\ 0.2714 \\ 0.3207
\end{bmatrix}$
&
$\begin{bmatrix}
1.81 \\ 13.81 \\ 83.31 \\ 74.81
\end{bmatrix}$
&
$\begin{bmatrix}
Yes \\ No \\ No \\ No
\end{bmatrix}$
&
$\begin{bmatrix}
0.392 \\ 0.433 \\ 2.477 \\ 2.773
\end{bmatrix}$
\end{tabular}
\end{table}
Here we see that the smallest investments in the Markowitz portfolio are not necessarily the stocks to remove (to achieve the best sparse portfolio).  The best portfolio with single stock is stock 1.  The best portfolio of size 2 contains Stock 3 and 4.  The best portfolio of size 3 excludes stock 1. The Relative Sparsity Costs seem to hint at many of those choices.

\section{Computational Results}
\label{sec:computationalresults}
\medskip

\subsection{Data, Parameters and Models}
\setcounter{equation}{0}
\label{sec:dataandparameters}

To test the $\ell_p$-norm regularized models, we collected historical daily stock price data in S \& P 500 index from CRSP Database\footnote{We choose this short time-interval due to the need for a large number of intervals and the common belief that the distribution of stock prices fundamentally change shape over decades.}, which spans from 31/12/2007 to 31/12/2012.  We don't include any company unless it is traded on the market
at least 90\% of the trading days during the data period, nor do any company not listed on the market for the entire timescale. The total list has 461 companies by 1259 trading days. Since S \& P 500 stocks have a high average correlation around 0.4516, for the purpose of
testing our model under more uncorrelated data, we further considered a larger dataset that contains 53 commodity ETF daily data from American market, and 236 stocks data of Husheng 300 Index  from Chinese market.\footnote{This index contains 60\% of the market value of stocks listed in Shanghai and Shengzheng Stock Exchange of China.} To deal with the mismatch between China and America's calendars,  we set the return of stocks not traded
because of holidays on either country to zero. We employ the rolling-window method to evaluate the out-of-sample performance\footnote{Taking account into the computational time, we use 36 rolling-windows for No-shorting Constraint case and Shorting-allowed $\ell_p$-norm model, $\ell_1$-norm ball constrained  model, 12 rolling-window for $\ell_1$-norm ball constrained $\ell_p$-norm regularization model and $\ell_p-\ell_2$-norm double regularization model}, with 500 days and 537 days training window, 21 days and 63 days estimation window respectively. The portfolios obtained from S\&P data and International data are named as S\&P Portfolio and International Portfolio, respectively.

Note that the coefficient $c=\phi m$ in the linear objective term of the regularized models. To solve the $\ell_p$-norm Markowitz models,
 proper $\phi$ values should be be chosen accordingly. To achieve this objective, we first set reasonable values for the minimum target return $m_0$, and then calculate the $\phi$-values from the dual variables of  the models in constraint form.
We use mean, variance and Sharpe Ratio to evaluate the out-of-sample performance, where the Sharpe ratio computed  here uses the same method as \cite{DeMiguel:2009}.

%
%
%

\subsection{No-shorting Constraint Case}
In \cite{Demiguel:2009a}, the authors apply the $\ell_1$-norm technique to seek sparse portfolios. The $\ell_1$-norm, however,  plays no role in the Markowitz model with no-shorting constraints. However, since no-shorting environments and investors exist extensively in the real market, we turn to the $\ell_p$-norm regularization to seek portfolios with desired sparsity in this situation. As we will see later, our $\ell_p$-norm regularized model \eqref{MV-lp} with no-shorting constraints produces extremely sparse portfolios with comparison to the already sparse Markowitz no-shorting model portfolios.

The $\ell_p$-norm regularized model is compared with two benchmarks in the framework of Markowitz model with no-shorting constraints. The first one is the Markowitz model without regularization ($\lambda=0$) and the second is the cardinality-constrained portfolio selection (CCPS) model. The global optimal cardinality-constrained portfolios are found by solving the following integer formulation of problem \eqref{MV-cardi}:
\[ \begin{array}{rl}
\min &\displaystyle  \frac{1}{2}x^TQx - c^Tx\\[0.2cm]
\mbox{s.t.}       &  e^Tx =1,\\[0.1cm]
& 0 \leq x \leq d,\\[0.1cm]
&e^Td \leq K,\\[0.1cm]
& d \in \{0,1\}^n.
\end{array}
\]

\subsubsection{In-Sample Performance}

Table \ref{tab:unregulated} reports the portfolio weight, the mean, the variance and the sparsity of  the  Markowitz portfolios with the specified return $m_0$ ranging from 0.02\% to 0.12\%.  The portfolios range from 19 to 26 stocks, which are about 4.1\%-5.6\% of the full set. The expected return of each portfolio equals or exceeds the minimum target return. The trend that portfolios with higher target return also have higher estimated variance is clear in the table.

Table \ref{tab:regulated} lists the results of the $\ell_p$-norm regularized Markowitz model
with $\lambda =5.5e-6$ by our second-order interior point algorithm. Clearly,  the resulting portfolios are
of low variance and larger sparsity. Specifically, the number of positive position ranges from
3 to 6,  which are only 15-25\% of the number of stocks  in the  Markowitz portfolios and 0.5-1.5\% of the total number of stocks. We also find that these portfolios have a similar composition to the non-zero unregularized counterparts. The top companies are the same (SO, K, KMB, GIS, AZO) and there is a complete overlap between the unregularized and regularized models: none of the companies in the sparse portfolios were found with 0\% stake in the unregularized portfolios.  However, 
\newpage
\begin{table}[h]
\caption{Results for the Unregularized  Markowitz Model: In-Sample Performance. }
\label{tab:unregulated}
\centering
\vskip 0.1cm
\tiny
\begin{tabular}{ c c c c c }
\begin{tabular}{| c | c |}
\hline
$\lambda$ & 0 \\
$m_0$ & 0.0002 \\
$\phi$ & 0.00000 \\
\hline
Mean & 0.00029\\
Variance & 3.53e-5\\
Sparsity & 19\\
\hline
SO & 0.29541 \\
K & 0.13471 \\
GIS & 0.08696 \\
KMB & 0.07952 \\
PEP & 0.07762 \\
AZO & 0.07409 \\
WMT & 0.0475 \\
MCD & 0.04266 \\
HSY & 0.03874 \\
NEM & 0.03364 \\
CPB & 0.032 \\
PG & 0.026 \\
SYY & 0.01039 \\
FOH & 0.0093 \\
NFLX & 0.00565 \\
PPL & 0.00523 \\
JNJ & 0.00003 \\
CAG & 0.00002 \\
ABT & 0.00002 \\

\hline
\end{tabular}

& &

\begin{tabular}{| c | c |}
\hline
$\lambda$ & 0 \\
$m_0$ & 0.0004 \\
$\phi$ & 0.00445 \\
\hline
Mean &  0.0004\\
Variance &  3.89e-5\\
Sparsity & 24\\
\hline
SO & 0.27974 \\
KMB & 0.12563 \\
K & 0.12469 \\
AZO & 0.0961 \\
GIS & 0.07693 \\
HSY & 0.07123 \\
PEP & 0.06607 \\
WMT & 0.05964 \\
NEM & 0.02297 \\
MCD & 0.01736 \\
AAPL & 0.01388 \\
CPB & 0.0132 \\
PG & 0.01302 \\
CAG & 0.00545 \\
FOH & 0.00482 \\
PPL & 0.00302 \\
DUK & 0.00246 \\
SYY & 0.0017 \\
REGN & 0.00081 \\
ABT & 0.0007 \\
JNJ & 0.00003 \\
ORLY & 0.00002 \\
MNST & 0.00001 \\
SHW & 0.00001 \\

\hline
\end{tabular}

& &

\begin{tabular}{| c | c |}
\hline
$\lambda$ & 0 \\
$m_0$ & 0.0006 \\
$\phi$ & 0.01099 \\
\hline
Mean & 0.0006\\
Variance &  3.89e-5\\
Sparsity & 26\\
\hline
SO & 0.25522 \\
KMB & 0.17225 \\
AZO & 0.11016 \\
HSY & 0.10648 \\
K & 0.09903 \\
GIS & 0.06227 \\
WMT & 0.06134 \\
PEP & 0.0321 \\
AAPL & 0.02935 \\
REGN & 0.01786 \\
SHW & 0.01635 \\
CAG & 0.01589 \\
ABT & 0.00896 \\
DUK & 0.00846 \\
NEM & 0.00209 \\
THC & 0.00094 \\
GILD & 0.00059 \\
PPL & 0.00005 \\
PG & 0.00004 \\
ORLY & 0.00003 \\
MNST & 0.00003 \\
CPB & 0.00002 \\
LLY & 0.00002 \\
BIIB & 0.00001 \\
RAI & 0.00001 \\
SYY & 0.00001 \\

\hline
\end{tabular}

\\ \\ \\

\begin{tabular}{| c | c |}
\hline
$\lambda$ & 0 \\
$m_0$ & 0.0008 \\
$\phi$ & 0.01886 \\
\hline
Mean &  0.0008\\
Variance &  4.48e-5\\
Sparsity & 21\\
\hline
SO & 0.21501 \\
KMB & 0.20186 \\
HSY & 0.13914 \\
AZO & 0.11891 \\
K & 0.05701 \\
WMT & 0.04926 \\
SHW & 0.04633 \\
AAPL & 0.03462 \\
REGN & 0.03401 \\
GIS & 0.03158 \\
CAG & 0.02365 \\
BIIB & 0.01816 \\
DUK & 0.01522 \\
GILD & 0.00742 \\
THC & 0.00495 \\
ABT & 0.00158 \\
MNST & 0.00088 \\
ORLY & 0.00004 \\
LLY & 0.00002 \\
PEP & 0.00001 \\
RAI & 0.00001 \\

\hline
\end{tabular}

& &

\begin{tabular}{| c | c |}
\hline
$\lambda$ & 0 \\
$m_0$ & 0.0010 \\
$\phi$ & 0.02731 \\
\hline
Mean &  0.001\\
Variance &  5.4e-5\\
Sparsity & 23\\
\hline
KMB & 0.22204 \\
HSY & 0.16869 \\
SO & 0.15838 \\
AZO & 0.12639 \\
SHW & 0.0764 \\
REGN & 0.05116 \\
BIIB & 0.04093 \\
AAPL & 0.03795 \\
WMT & 0.03061 \\
CAG & 0.02457 \\
DUK & 0.02241 \\
GILD & 0.01254 \\
MNST & 0.01105 \\
THC & 0.00928 \\
K & 0.00699 \\
ORLY & 0.00008 \\
GIS & 0.00004 \\
C & 0.00004 \\
ABT & 0.00003 \\
HD & 0.00002 \\
LLY & 0.00001 \\
RAI & 0.00001 \\
AMGN & 0.00001 \\

\hline
\end{tabular}

& &

\begin{tabular}{| c | c |}
\hline
$\lambda$ & 0 \\
$m_0$ & 0.0012 \\
$\phi$ & 0.03645 \\
\hline
Mean &  0.0012 \\
Variance &  6.68e-5\\
Sparsity & 20\\
\hline
KMB & 0.22846 \\
HSY & 0.18915 \\
AZO & 0.12588 \\
SHW & 0.10123 \\
SO & 0.08532 \\
REGN & 0.06975 \\
BIIB & 0.05887 \\
AAPL & 0.03929 \\
DUK & 0.03077 \\
HD & 0.0265 \\
GILD & 0.014 \\
THC & 0.0135 \\
MNST & 0.01165 \\
CAG & 0.00304 \\
C & 0.00203 \\
WMT & 0.00019 \\
ORLY & 0.00002 \\
EXPE & 0.00002 \\
AMGN & 0.00001 \\
ABT & 0.00001 \\

\hline
\end{tabular}
\end{tabular}
\end{table}

\noindent the composition is far from identical as many low-weighted stocks in the unregularized portfolios have large weights in the sparse portfolios. Moreover, Figure \ref{fig:portfoliosparsity} shows the number of positive positions versus
the regularization parameter $\lambda$ graph of the  $\ell_p$-norm regularized portfolios.
With minor exception, increasing lambda almost always results in a more sparse
solution which is consistent with our portfolio theory developed in Section \ref{sec:theory}.
\newpage

\begin{table}[t]
\caption{Results for $\ell_p$-norm Markowitz Regularized Portfolios with  $\lambda=5.5e-6$: In-Sample Performance.}
\label{tab:regulated}
\centering
\tiny
\begin{tabular}{ c c c c c }
\begin{tabular}{| c | c |}
\hline
$\lambda$ & 5.5e-6 \\
$m_0$ & 0.0002 \\
$\phi$ & 0.00000 \\
\hline
Mean &  0.00025\\
Variance & 4.12e-5 \\

Sparsity & 3\\
\hline
SO & 0.55327 \\
K & 0.23209 \\
GIS & 0.21464 \\

\hline
\end{tabular}

& &

\begin{tabular}{| c | c |}
\hline
$\lambda$ & 5.5e-6 \\
$m_0$ & 0.0004 \\
$\phi$ & 0.00445 \\
\hline
Mean &  0.00025\\
Variance &  4.12e-5\\
Sparsity & 3\\
\hline
SO & 0.55222 \\
K & 0.22999 \\
GIS & 0.21779 \\

\hline
\end{tabular}

& &

\begin{tabular}{| c | c |}
\hline
$\lambda$ & 5.5e-6 \\
$m_0$ & 0.0006 \\
$\phi$ & 0.01099 \\
\hline
Mean &  0.00046\\
Variance &  4.05e-5\\
Sparsity & 4\\
\hline
SO & 0.38804 \\
KMB & 0.31953 \\
K & 0.16675 \\
HSY & 0.12568 \\

\hline
\end{tabular}

\\ \\ \\

\begin{tabular}{| c | c |}
\hline
$\lambda$ & 5.5e-6 \\
$m_0$ & 0.0008 \\
$\phi$ & 0.01886  \\
\hline
Mean & 0.00061 \\
Variance & 4.35e-5 \\
Sparsity & 4 \\
\hline
KMB & 0.35527 \\
SO & 0.27413 \\
HSY & 0.20689 \\
AZO & 0.16371 \\

\hline
\end{tabular}

& &

\begin{tabular}{| c | c |}
\hline
$\lambda$ & 5.5e-6 \\
$m_0$ & 0.0010 \\
$\phi$ & 0.02731 \\
\hline
Mean &  0.00098\\
Variance &  5.76e-5\\
Sparsity & 5 \\
\hline
KMB & 0.47565 \\
HSY & 0.2479 \\
AZO & 0.19227 \\
REGN & 0.06476 \\
DUK & 0.01942 \\

\hline
\end{tabular}

& &

\begin{tabular}{| c | c |}
\hline
$\lambda$ & 5.5e-6 \\
$m_0$ & 0.0012 \\
$\phi$ & 0.03645 \\
\hline
Mean &  0.00109\\
Variance & 6.41e-5 \\
Sparsity & 6\\
\hline
KMB & 0.43308 \\
HSY & 0.25649 \\
AZO & 0.18254 \\
REGN & 0.08895 \\
DUK & 0.02721 \\
THC & 0.01172 \\

\hline
\end{tabular}
\end{tabular}
\end{table}

\begin{figure}[t]
\begin{center}
\includegraphics[height=2in]{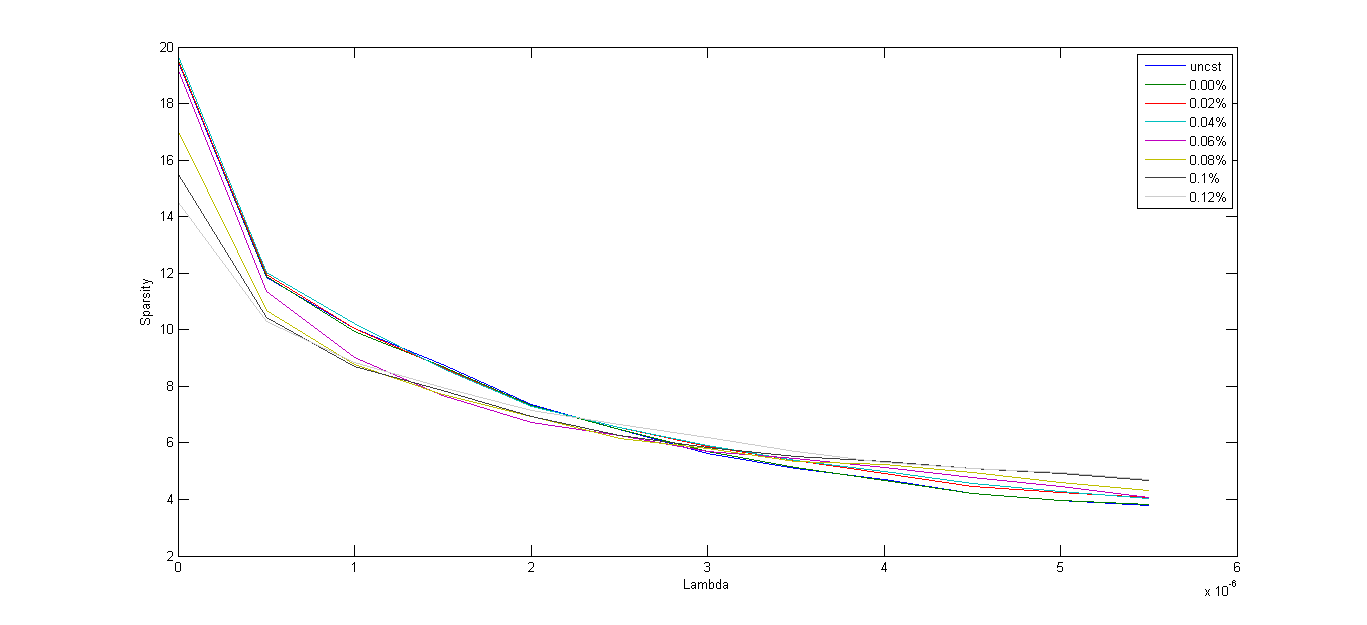}
\caption{Portfolio Sparsity}
\label{fig:portfoliosparsity}
\end{center}
\end{figure}

A comprehensive comparison of computational results between our $\ell_p$-norm regularized model and the cardinality constrained portfolio selection (CCPS) model are reported in the Table \ref{tab:No Shorting case}. As can be seen in the table, our regularized $\ell_p$-norm performs
almost as well as theoretical possible---the difference of the variance estimation between the two models are
within 0.2\% in all cases and the difference of the mean estimation are within 0.02\%. Therefore, compared to the computational intractable cardinality constrained portfolio optimization,
our $\ell_p$-norm regularized portfolio, which can be obtained in polynomial time, performs almost as well and seeks near optimal sparse portfolios.

\begin{table}[t]
 \renewcommand{\arraystretch}{1.2}
\centering
\caption{Comparison of Sparsity, Mean and Variance between the $\ell_p$-norm  Model and  CCPS (Daily Return).}
\vskip 0.15cm
\label{tab:No Shorting case}
\scriptsize
    \begin{tabular}{c c c c c c c c c}
\hline
~   & \multicolumn{4}{c}{$\ell_p$-norm} & ~ &\multicolumn{3}{c}{CCPS} \\ \cline{2-5} \cline{7-9}
~  & $\lambda$ & Sparsity & Mean & Variance &~ & Sparsity & Mean & Variance \\ \hline
$m_0 = 0.02\% \quad$ & 5.0e-7 & 9   & 0.05\%  & 4.45\% &~ &  9 & 0.05\% &4.46\% \\
~ & 1.0e-6 & 7   & 0.03\% & 3.68\% & ~ &7 & 0.04\% & 3.66\% \\
~& 2.0e-6 & 5 & 0.03\% & 3.90\% & ~ & 5 & 0.04\% & 3.75\% \\
~ &3.5e-6&4& 0.02\% & 4.08\% & ~ & 4 & 0.04\% & 3.89\% \\
~ &4.5e-6&3& 0.02\% & 4.12\% & ~ & 3& 0.04\% & 4.06\% \\
~&~&~&~&~&~&~&~&~ \\
$m_0 = 0.1\%\quad $ & 5.0e-7 & 10   & 0.06\%  & 4.57\% &~ &  10 & 0.05\% &4.58\% \\
~ & 1.0e-6 & 7   & 0.10\% & 4.90\% & ~ &7 & 0.10\% & 4.86\% \\
~& 2.0e-6 & 7 & 0.09\% & 5.27\% & ~ & 5 & 0.09\% & 5.37\% \\
~ &3.5e-6&6& 0.09\% & 5.18\% & ~ & 6 & 0.09\% & 5.28\% \\
~ &4.5e-6&6& 0.09\% & 5.18\% & ~ & 6& 0.09\% & 5.28\% \\ \hline
\end{tabular}
\end{table}

\subsubsection{Out-of-Sample Performance}
 \cite{Brodie:2009} show that sparse portfolios  are often more robust and thus  outperform the portfolios with less sparsity  in terms of out-of-sample performance. In their analysis,  the no-shorting constraint ($x \ge 0$) is taken as the \emph{most} extreme sparsity inducing measure. We continued this investigation by taking the no-shorting constraint as the \emph{least} extreme measure and adding  the  $\ell_p$-norm regularizer
onto the objective function. It is interesting to ask whether the sparsest portfolios will outperform other portfolio strategies with less sparsity.

Figure \ref{fig:portfolioreturns} and Figure \ref{fig:portfoliovariances} show the out-of-sample portfolio returns and variances obtained by the $\ell_p$-norm regularized Markowitz model with $\lambda$ ranging from 5.0e-7 to 5.5e-6 and the CCPS, respectively. From Figure \ref{fig:portfolioreturns}, we observe clearly that  most of the plots go up slightly and then achieve its maximum, indicating that the portfolios with moderate sparsity (around 10) perform very well, even better than the Markowitz portfolio. However, with the continuously increasing of sparsity, the mean will go down dramatically and thus the regularized portfolios with extreme sparsity perform poorly in the sense of portfolio mean. Figure \ref{fig:portfoliovariances} shows that the variance of the regularized portfolios is increasing with a incremental rate with the increasing sparsity of the portfolios. However, though the highly sparse portfolios performs poorly in the sense of portfolio variance, the intermediate portfolios with about 10 companies suffered a 15-25\% increase in variance which is also comparable to the CCPS integer portfolios.

\begin{figure}[t]
\begin{center}
\includegraphics[height=2in]{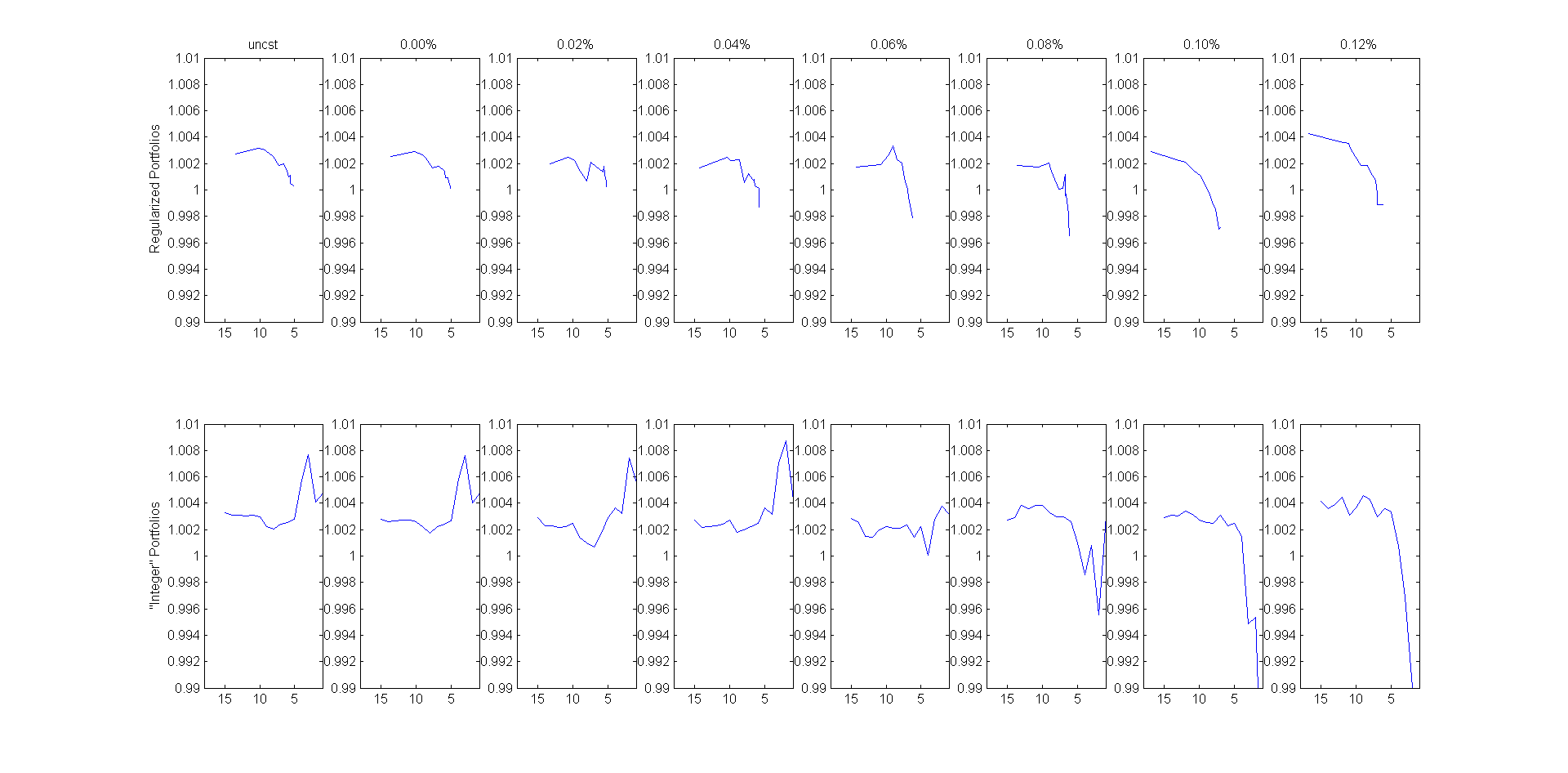}
\caption{Portfolio Returns}
\label{fig:portfolioreturns}
\end{center}
\end{figure}

\begin{figure}[t]
\begin{center}
\includegraphics[height=2in]{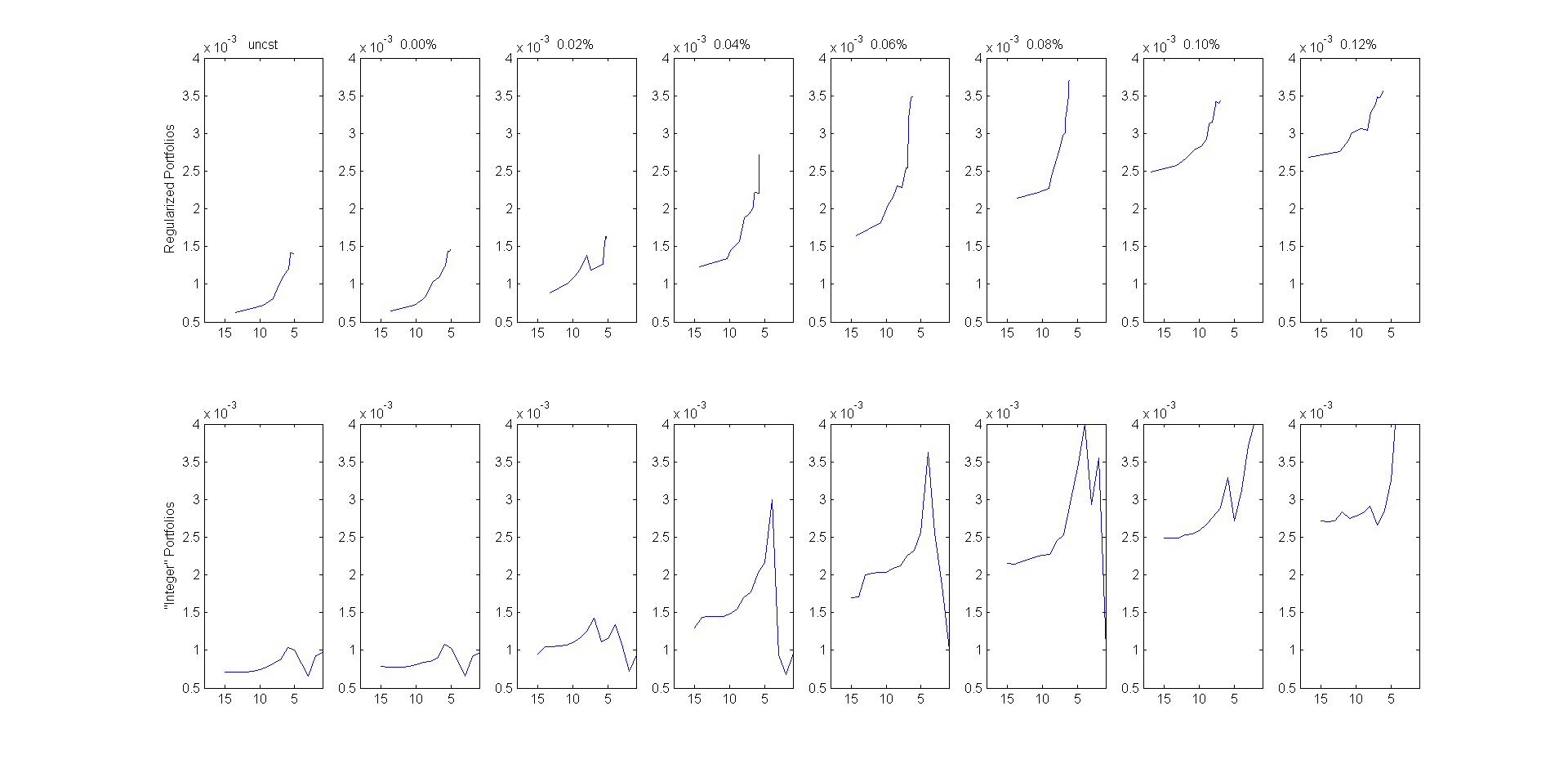}
\caption{Portfolio Variances}
\label{fig:portfoliovariances}
\end{center}
\end{figure}

\begin{figure}[t]
\begin{center}
\includegraphics[height=2in]{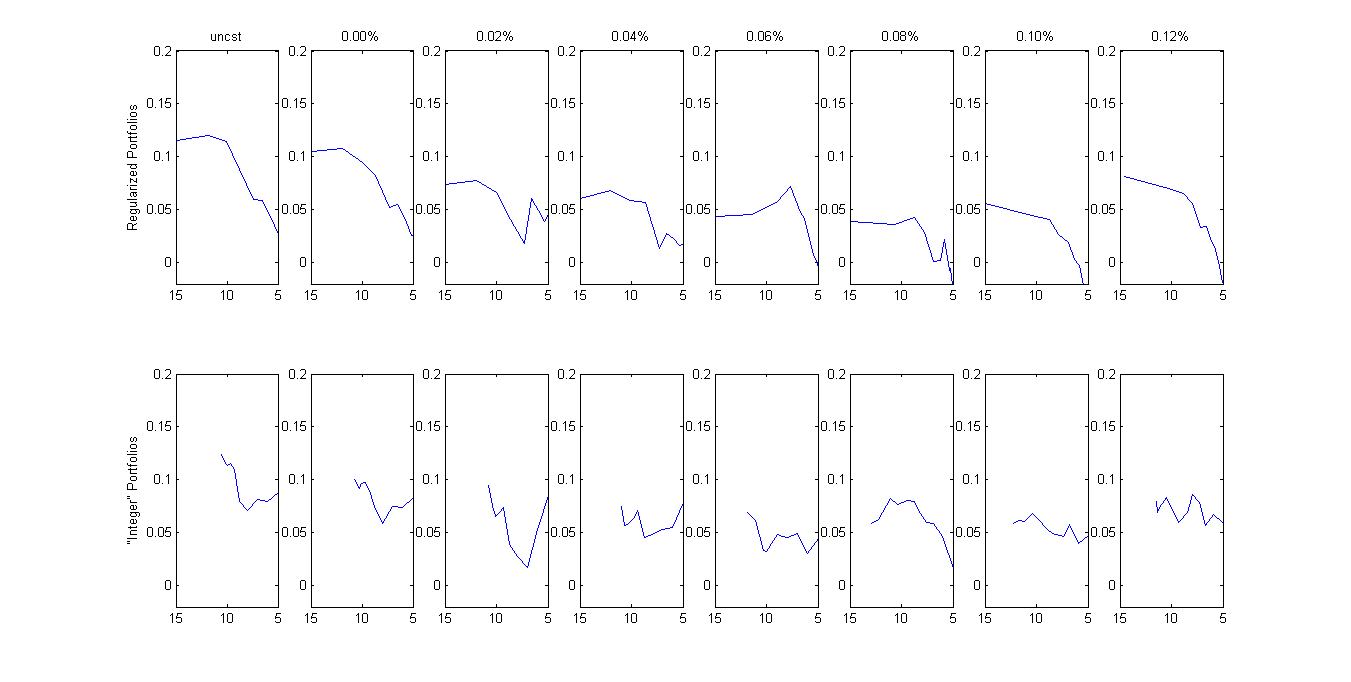}
\caption{Portfolio Sharpe Ratios}
\label{fig:sharp_stock}
\end{center}
\end{figure}

Figure \ref{fig:sharp_stock} shows the out-of-performance Sharpe ratios of our  $\ell_p$-norm
regularized portfolio and the CCPS integer portfolio. Although the Markowitz portfolio (with $\lambda=0$) outperforms our $\ell_p$-norm regularized model in terms of the out-of-sample Sharpe ratio, the sparse portfolios may be more implementable due to  the  transaction costs or logistical limitations reasons. Our results indicate that an intermediate sparse portfolio may get a comparable or at most only 10-20\% cost in Sharpe ratio while reducing more construction costs. Also, the $\ell_p$-norm regularized approach is competitive with the computationally gigantic integer approach in the sense of
out-of-sample performance.

\subsection{Shorting-Allowed Extension}
\label{sec:shortingallowedextension}

Next we relaxed our constraint to allow the short-selling of stocks. We compare our model \eqref{MV-lp-short} with the $\ell_1$-norm ball constrained portfolios studied by \cite{Demiguel:2009a}, as the strategy may find sparse portfolios with improved out-of-sample Sharpe ratios.

\subsubsection{ $\ell_p$-norm Regularized Model}
\label{sec:$p$-norm}
Figure \ref{fig:lambdasparsity}  shows that the shorting-allowed Markowitz portfolios behave eccentrically (also see Table \ref{tab:Shorting case}), with the portfolio including all the stocks no matter the choice of  the parameter $\phi$.  Meanwhile, our $\ell_p$-norm regularized model \eqref{MV-lp-short} is able to reduce the number of investing stocks drastically. For example, only 22 stocks are involved in the Markowitz regularized portfolio for $\lambda =1e-6$ and $m_0=0.06\%$, and thus there is a 95.2\% reduction of the portfolio size.  The parameter $\lambda$ can be regarded as a server to control the portfolio sparsity.

\begin{figure}[t]
\begin{center}
\includegraphics[height=2in]{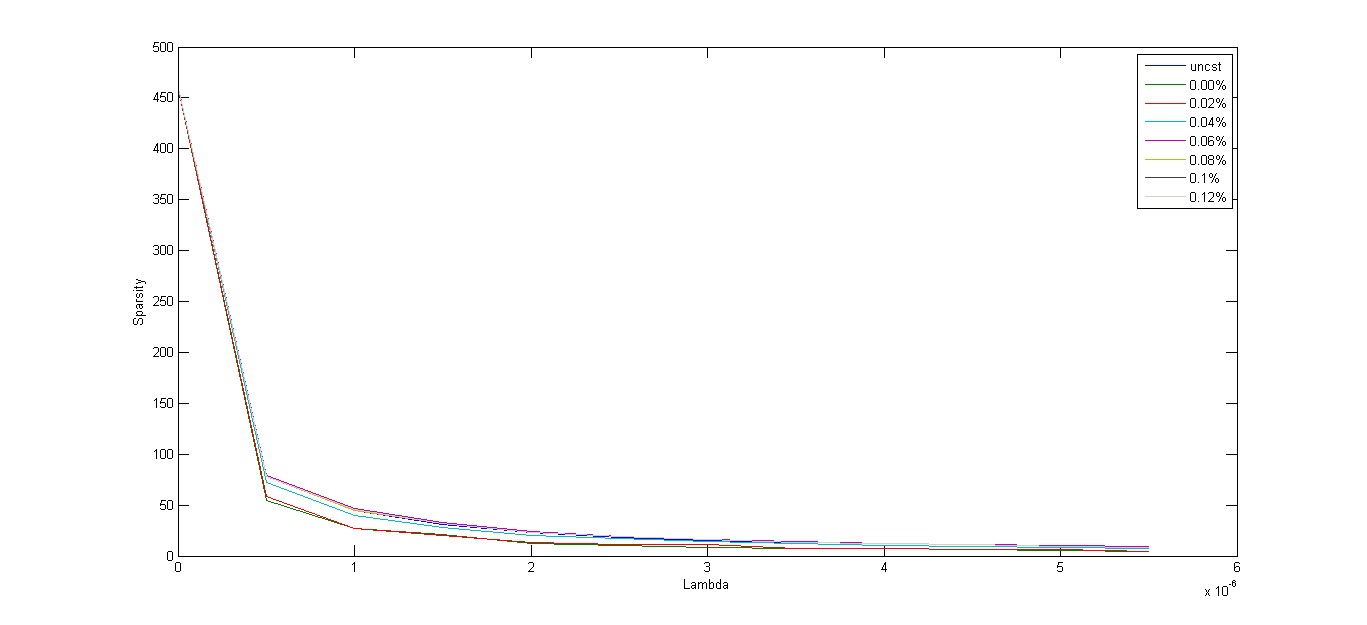}
\caption{Portfolio Sparsity}
\label{fig:lambdasparsity}
\end{center}
\end{figure}

The out-of-sample results are similar to the shorting-prohibited case.  From  Table \ref{tab:Shorting case}, we see that the Sharpe ratio tends to be the highest when $\lambda$ is not too large, and would decrease with the increasing of the parameter $\lambda$. However, even for a significantly small $\lambda$, the regularized portfolios are much more sparse (e.g. 79 versus 461), and of competitive or better performance  while compared with  the Markowitz portfolio.  For larger values of $\lambda$, there is a clear tradeoff between the  portfolio sparsity and performance.
\begin{table}[t]
 \renewcommand{\arraystretch}{1.2}
\centering
\vskip 0.15cm
\caption{Sharpe Ratio and Sparsity of shorting allowed $\ell_p$-norm regularized Model}
\label{tab:Shorting case}
\scriptsize
    \begin{tabular}{c c c c c c c c c c c c}
\hline
~ &  \multicolumn{2}{c}{$m_0 =-\infty$}  &  ~ &  \multicolumn{2}{c}{$m_0=0.02\%$} & ~ &
\multicolumn{2}{c}{$m_0 =0.06\%$} & ~ &   \multicolumn{2}{c}{$m_0 =0.1\%$}\\
\cline{2-3}\cline{5-6}\cline{8-9}\cline{11-12}
$\lambda$ & Spa & SRatio & ~ & Spa & SRatio & ~ & Spar & SRatio & ~ & Spar & SRatio\\ \hline
0	        &	461.0 	&	0.165        & ~ &	461.0 	&	0.161 	&~&	461.0 	&	0.146 	&~&	461.0 	&	0.127 	\\
5.0e-7	&	78.1 	&	0.156 	& ~ &	58.7  	&	0.161 	&~&	79.3 	&	0.161 	&~&	78.2	       &	0.166 	\\
1.0e-6	&	45.1 	&	0.125 	& ~ &	27.1 	&	0.120 	&~&	46.9 	&	0.123 	&~&	45.2 	&	0.120 	\\
2.0e-6	&	22.9   	&	0.159 	& ~ &	13.4 	&	0.159 	&~&	24.4 	&	0.159 	&~&	23.4 	&	0.155 	\\
2.5e-6	&	18.4	       &	0.149 	& ~ &	11.5 	&	0.152 	&~&	19.2 	&	0.147 	&~&	18.8 	&	0.150 	\\
3.5e-6	&	13.4 	&	0.120 	& ~ &	7.6 	        &	0.120 	&~&	14.3 	&	0.118 	&~&	13.4 	&	0.121 	\\
4.5e-6	&	10.9 	&	0.040 	& ~ &	6.6 	        &	0.036 	&~&	11.0	       &	0.040 	&~&	10.8 	&	0.041 	\\
5.5e-6	&	  8.5 	&	0.024 	& ~ &	4.9 	       &	0.023 	&~&	8.9	       &	0.024 	&~&	8.6     	&	0.027 	\\
\hline
\end{tabular}
\end{table}

\begin{table}[t]
 \renewcommand{\arraystretch}{1.2}
\centering
\caption{Sharpe Ratio and Sparsity of $\ell_1$-norm ball constrained Markowitz portfoliol}
\vskip 0.1cm
\label{tab:l1norm}
\scriptsize
    \begin{tabular}{c c c c c c c c c c c c}
\hline
~ &  \multicolumn{2}{c}{$m_0 =-\infty$}  &  ~ &  \multicolumn{2}{c}{$m_0=0.02\%$} & ~ &
\multicolumn{2}{c}{$m_0 =0.06\%$} & ~ &   \multicolumn{2}{c}{$m_0 =0.1\%$}\\
\cline{2-3}\cline{5-6}\cline{8-9}\cline{11-12}
$\delta$ & Spa & SRatio & ~ & Spa & SRatio & ~ & Spar & SRatio & ~ & Spar & SRatio\\ \hline										
1.5	&70.5	&	0.127 &~&	70.4	 &	0.109	&~&	68.3	      &	0.111 	&~&	60.9	&	0.149	\\
2	&118.1	&	0.163 &~&	118.6&	0.155 	&~&	116.3	&	0.177 	&~&	111.1	&	0.181	\\
\hline
\end{tabular}
\end{table}

\subsubsection{$\ell_1$-norm Ball Constrained Model}
\label{sec:$1$-norm}

For the purpose of comparison, we also post the results of $\ell_1$-norm ball constrained portfolios on the same data set. The
$\ell_1$-norm ball constrained model considered in this section takes the following form
\begin{equation}\label{eq:markowitzwith1normconstraint}
\begin{array}{rl}
\min & \frac{1}{2}x^TQx \\[0.2cm]
\mbox{s.t.}       &  e^Tx =1, \\[0.2cm]
 &  m^Tx \ge m_0, \\[0.2cm]
  &  \|x\|_1 \le \delta,\\[0.2cm]
\end{array}
\end{equation}
where $\delta \geq 1$. Figure \ref{fig:PortfoliosparsitywithDifferentDelta} shows the number of nonzero positions versus
the threshold parameter $\delta$  of the  $\ell_1$-norm ball constrained portfolios. It is clear
that the sparsity decreases at a fast speed with the increasing of  $\delta$. The left-most date points
corresponds to the shorting-allowed Markowitz model where the $\ell_1$-norm ball constraint is not effective while the right-most
data points ($\delta =1$) corresponds to the shorting-prohibited Markowitz mode where the $\ell_1$-norm ball constraints takes its
most effective role in inducing sparsity. However,  even when $\delta = 1$, the average sparsity of the portfolio is around 15, which is much more dense than the shorting-allowed portfolio with $\lambda$ smaller than 3.5e-6.

\begin{figure}[t]
\begin{center}
\includegraphics[height=2in]{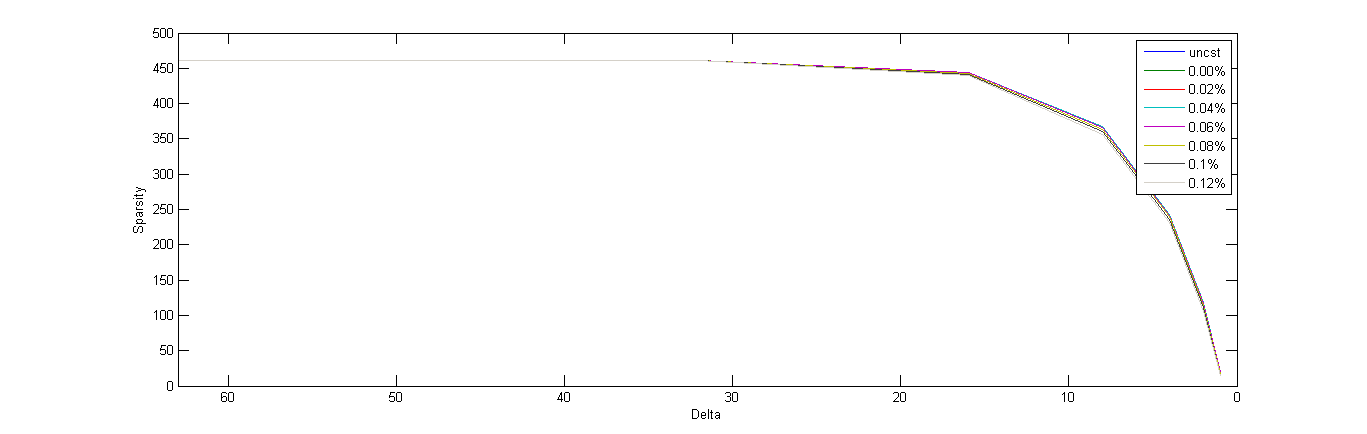}
\caption{Portfolio Sparsity with Different Delta of $\ell_1$-Norm Model}
\label{fig:PortfoliosparsitywithDifferentDelta}
\end{center}
\end{figure}

We can make  a detailed comparison between  $\ell_1$-norm ball constrained model and  $\ell_p$-norm regularization model by Table \ref{tab:Shorting case} and  \ref{tab:l1norm}.  It is easy to see that when $\delta$ equals 1.5 or 2, the out-of-sample performance of $\ell_p$-norm
models is similar to that of the $\ell_1$-norm ball constrained portfolios but the former is much more sparse.  However, the performance of the $\ell_p$-norm models is surpassed when $\delta$ is increased.  In that case,  the $\ell_1$-norm ball constrained portfolio  achieves a better out-of-sample performance with the sacrifice of sparsity, see Figure \ref{fig:PortfolioSharperatiowithsparsity}. Moreover, the largest out-of-sample Sharpe Ratio is achieved when $\delta \approx 16$. From this figure, we can also see that the shorting-allowed Markowitz Model (far left points) is  better than shorting-prohibited Markowitz Model (far right points).  This is consistent with the remark made by \cite{Jagannathan:2003} that when daily data is used, shorting-prohibited models perform almost as well.

\begin{figure}[t]
\centering
\includegraphics[scale=.36]{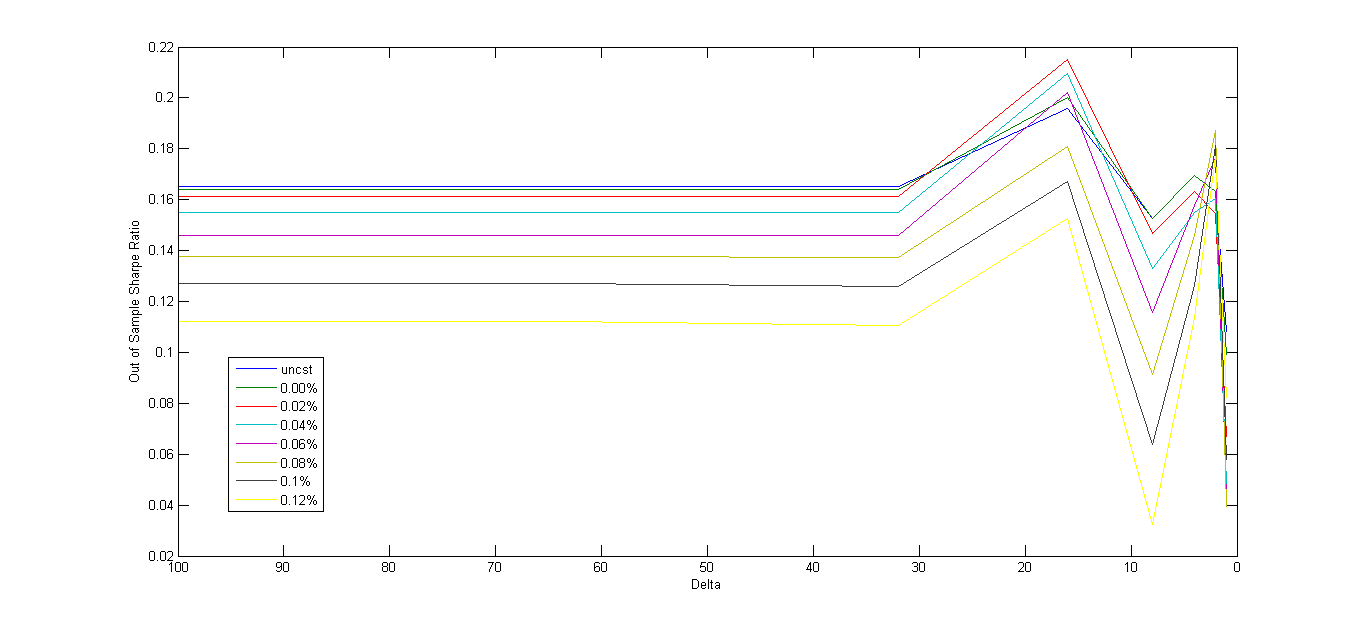}
\caption{Portfolio Sharpe Ratio with Sparsity}
\label{fig:PortfolioSharperatiowithsparsity}
\end{figure}

\subsection{$\ell_1$-norm Ball Constrained $\ell_p$-norm Regularized Model}

In the last two sections, we have discussed the  computational performance of the  $\ell_1$-norm ball constrained Markowitz model and the shorting-allowed $\ell_p$-norm regularized Markowitz models \eqref{MV-lp} and \eqref{MV-lp-short} individually. Next, we consider the $\ell_1$-norm ball constrained $\ell_p$-norm regularized model  \eqref{MV-lpreg-l1con} to investigate the relationship between the leverage (characterized by the $\ell_1$-norm), the sparsity (induced mostly by the  $\ell_p$-norm) and the out-of-sample performance. According to the results of the $\ell_1$-norm constrained model, we
solve our regularized model combined with the $\ell_1$-norm constraint with $\delta$ ranging from
 from 1.5 to 32 and $(\lambda,\phi)$ taking an array of values. This thorough approach are expected to give
 us a more structured picture of the relationship between the $\ell_1$- and $\ell_p$-norms as well as their relationship to the performance.

 Table  \ref{tab:ssratio}  reports the out-of-sample computational results for the cases where $\delta$ is taken as
 1.5, 2 and 32, and $m_0$ is set as $-\infty$ and 0.04\%. Since the $\ell_1$-norm constrained $\ell_p$-norm regularization enjoys the similar trend for different choices of $\delta$, we don't report the corresponding results for succinctness. From the table, we see clearly that the  sparsity, in general, is antagonistic to performance.  Thus, there exists a tradeoff between performance and sparsity. Though the performance varies for different values of $\lambda$ , a well performed portfolio can be obtained when $\lambda $ is smaller than 4.5e-06.
And when $\lambda$ is not very large, say less then $2e-6$, the sparsity need not come at a high price (of Sharpe ratio)
and there are many sparse portfolios with comparable performance to the portfolios found with $\lambda=0$. Also, we  find that with the increase of $\lambda$, the leverage of the resulting portfolio decreases significantly. Thus, it seems that the leverage of the portfolio is mostly determined by the choice of $\lambda$. Moreover,  we note that there appears to be little cross-effect between sparsity and leverage on performance.

\begin{table}[t]
 \renewcommand{\arraystretch}{1.2}
\centering
\caption{ Sparsity and Sharpe Ratio of the Combined Model and $\ell_1$ Norm model for S \& P Data with Three Month Estimation Window: Out-of-Sample Performance}
\vskip 0.1cm
\label{tab:ssratio}
\scriptsize{
\begin{tabular}{c c c c c c c c c c c c c}
\hline
~   & ~ &\multicolumn{3}{c}{$\delta=1.5$} & ~ &\multicolumn{3}{c}{$\delta=2$} & ~ &\multicolumn{3}{c}{$\delta=32$}  \\ \cline{3-5} \cline{7-9}\cline{11-13}
model/ $m_0$  & $\lambda$ & Spar & Leve & SRatio &~ & Spar & Leve & SRatio &~ &Spar & Leve & SRatio \\ \hline
$\ell_1-\ell_p/-\infty$ & 5.0e-7 & 37.9 &1.496 	&0.074&~ &	54.3 	&	1.949 	&	0.159 	&~&	94.2 	&	3.392 	&	0.214 	\\\
~ & 2.0e-6 & 16.8 	&1.340 	&	0.101 	&~&	23.6 	&	1.526 	&	0.103 	&~&	29.7 	&	1.806 	&	0.063 	\\
~& 4.5e-6 &  9.3 	&	1.202 	&	0.120 	&~&	11.9 	&	1.285 	&	0.139 	&~&	13.8 	&	1.376 	&	0.050 	\\
~ &8.0e-6& 5.5 	&	1.084 	&	-0.094 	&~&	6.8 	&	1.130 	&	-0.026&~ 	&	6.8	&	1.155 	&	-0.138 	\\
~ &1.25e-5 &4.1 	&	1.049 	&	-0.094&~&	4.4 	&	1.063 	&	-0.132&~ 	&	4.4	&	1.068 	&	-0.111 \\
$\ell_1$-norm/$-\infty$ & --- & 70.5  &	1.500 	&	0.234 	&~&	117.2 	&	2.000 	&	0.195 	&~&	461	&	23.571 	&	0.247  \\ ~&~&~&~&~&~&~&~&~ \\
$\ell_1-\ell_p/0.04\%$  &5.0e-7	&	35.1 	&	1.496 	&	0.074 	&~&	56.6 	&	1.948 	&	0.177 	&~&	94.4 	&	3.414 	&	0.226 	\\
~ &2.0e-6	&	16.3 	&	1.338 	&	0.114 	&~&	23.6 	&	1.523 	&	0.115 	&~&	29.8 	&	1.808 	&	0.086 	\\
~ &4.5e-6	&	9.6    &	1.197 	&	0.105 	&~&	12.0 	&	1.284 	&	0.138 	&~&	13.6 	&	1.378 	&	0.053 	\\
~ &8.0e-6	&	5.6 	&	1.088 	&	-0.132&~ &  6.8 	&	1.129 	&	-0.006 &~	&6.7 	&	1.143 	&	-0.162 	\\
~ &1.25e-5	&	4.1 	&	1.049 	&	-0.098 &~	&4.4 	&	1.063 	&	-0.133 &~	&4.3 	&	1.068 	&	-0.115 	\\
$\ell_1$-norm/$0.04\%$ &--- & 70.2 	&	1.500 	&	0.201&~ 	&	118.7 	&	2.000 	&	0.207 	&~&	461.0 	&	23.500 	&	0.234 \\
\hline
\end{tabular}}
\end{table}

\subsection{$\ell_2-\ell_p$-norm Double Regularized Model}

As mentioned in \cite{DeMiguel:2009}, the $\ell_2$-norm constraint can be viewed as placing a prior on the 1/N strategy, thus it is reasonable to expect the results close to the 1/N strategy. Yet, most investors would not  invest into a portfolio with huge number of stocks, which  motivates us to
develop a portfolio strategy with less stocks  but similar to the 1/N strategy with competitive out-of-sample performance, especially for those passive investors. For this purpose, it is natural to consider the $\ell_p$-norm  regularization of the $\ell_2$-norm constrained Markowitz model
\[
 \label{eq:lpl2model}
\begin{array}{rl}
\min & \frac{1}{2}x^TQx  + \lambda\|x\|_p^p \\[0.2cm]
\mbox{s.t.}       &  e^Tx =1, \\[0.2cm]
 &  m^Tx \ge m_0, \\[0.2cm]
  &  \|x\|_2^2 \le \delta^2, \\[0.2cm]
\end{array}
\]
or its Lagranagin version (double regularization Markowitz model \eqref{MV-lp-l2}) to see  if we can obtain a portfolio that balances sparsity and uniform prior. The results of the $\ell_2-\ell_p$ model are shown in
Table \ref{tab: 2 and p combined sp}, with different choices of $\lambda$ and $\delta$. The parameters
$\mu$ and $\phi$ in the regularized model \ref{MV-lp-l2} are obtained from the dual variables of problem
\eqref{eq:lpl2model} with $\lambda =0$. Seen from the result,  the optimal portfolio obtained by
the double regularization formulation would include  all the stocks in the case that $\lambda =0$ for all values of $\delta$, closely related to the 1/N strategy. Also, the portfolio becomes more sparse
with the increasing of $\lambda$ and  fixed $\delta$, while more dense with the increasing of $\delta$
and fixed $\lambda$ . This trend shows a tradeoff between  $\ell_p$-norm regularization and $\ell_2$-norm
ball constraints.

 It is also note that the strategy to invest all stocks doesn't usually perform best  in the sense of Sharpe ratio. For example, in the case that $\lambda$=1.25e-5, $\delta$=0.1 and $m_0$=0.08\%, we can find a portfolio with only 135 stocks yet with a high Sharpe ratio 0.575, which is much better than the Sharpe ratio 0.374 attained with $\lambda=0$. Similar as the observation before, the extremely sparse portfolio often performs
poorly showing a tradeoff between sparsity and performance.

Also, the most constricting delta ($\delta=0.1$) had the highest performing portfolios, suggesting that the presence of a strong uniform prior on all stocks helps mitigate overfitting due to poor variance/covariance estimates.  The out-of-sample performance was \emph{increasing} in $\lambda$ when $\lambda$ was not too large.
These moderately sparse, highly $\ell_2$-norm constricted portfolios performed excellently (all had Sharpe Ratio near or above 0.5).  Thus the $\ell_2$ and $\ell_p$ norms appear to exhibit synergy in reducing overfitting.

Table \ref{tab: 2 and p combined Inter} lists the out-of-sample computational results of our $\ell_2-\ell_p$ double regularization model for international data with much more diversity. Compared with the results for S \& P data, the overall performance is greatly enhanced, especially for the sparsest portfolios. Very  surprisingly,
we even find that a portfolio with two stocks perform quite well. And also  we see that the cost of sparsity
need not be high even for very sparse portfolios if the stock base is favorable.

\begin{table}[t]
 \renewcommand{\arraystretch}{1.2}
\centering
\caption{Sparsity and Sharpe Ratio of the $\ell_2-\ell_p$-Norm Double-Regularization Model for S\& P Data with Three Month Estimation Window}
\vskip 0.15cm
\label{tab: 2 and p combined sp}
\scriptsize
    \begin{tabular}{c c c c c c c c c c c c c}
\hline
~   & ~ &\multicolumn{2}{c}{$\delta=0.1$} & ~ &\multicolumn{2}{c}{$\delta=0.2$} & ~ &\multicolumn{2}{c}{$\delta=0.3$}& ~ &\multicolumn{2}{c}{$\delta=0.4$} \\ \cline{3-4} \cline{6-7}\cline{9-10} \cline{12-13}
~ & $\lambda$ & Spar & SRatio&~ & Spar & SRatio & ~& Spar & SRatio&~& Spar &SRatio \\ \hline
$m_0 =0.00\%$ \quad &   0  & 461.0& 0.409 &~& 461.0 & 0.18 &~&461.0	& 0.087 & ~&461.0 & 0.035\\
 &  5.0e-7  & 321.8	&0.431 & ~ &185.8	&0.247&~&	133.0	&0.264	&~&	116.7&0.214\\
 & 2.0e-6 &236.8	&0.414&~	&	69.8	&0.35	& ~&	50.3	&0.29 &~ &33.3	&0.231\\
& 4.5e-6 &	164.4&	0.504	&~&	34.0	&0.422	&~&	34.3	&0.248&~	&15.8&0.14\\		
&8.0e-6 &166.6	&0.498&~	&	20.2	&0.336&~	&	13.5	&0.237&~	&	8.1& 0.096\\
&1.25e-5&105.3	&0.536&~	&	15.1	&0.288&~	&	7.6	&0.096&~	&	5.0 &	-0.083\\
\\
$m_0 =0.04\%$ \quad &   0  & 461.0	    &0.389	&~&	461.0 &  0.181&~	&	461.0	&0.085	&~    &461.0	&	0.046\\
 &  5.0e-7  &321.8   & 0.405	&~&	181.3	&   0.25&~	&	127.0 	&    0.236     &~	&111.8    &   0.233	\\
 & 2.0e-6 &233.3   &  0.385	&~&	63.8	&    0.374&~	&	46.2 	&   0.27        &~	&36.9	&  0.251\\	
& 4.5e-6 &189.6   &  0.414	&~&	34.8	&    0.43 &~	&	26.5	&  0.258	&~	&22.3 &  0.147\\
&8.0e-6 &188.9 	&  0.52	&~&	20.9	&   0.376&~	&	11.3	&  0.201	&~	&8.2  &  0.122\\
&1.25e-5&132.5	&  0.553	&~&	15.3	&  0.31&~	&	7.8	&  0.087	&~	& 4.9  & -0.082\\
\\
$m_0 =0.08\%$ \quad &   0  &460.9	&  0.374	&~&	461.0	 &  0.183	&~&	461.0	 &  0.089&~	&  461.0       &  0.052\\	
 &  5.0e-7  &343.7	&  0.376	&~&	198.3  &  0.286	&~&	144.7	&  0.246&~	&  120.6  &  0.243\\
 & 2.0e-6 &254.9	&  0.366	&~&	73.0	  &  0.39	&~&	63.4  	&  0.212&~	&  32.6    &  0.222\\
& 4.5e-6 &179.6	&0.379	&~&	35.4	&0.421	&~&	25.6	&0.288&~	&27.0&0.212\\
&8.0e-6 &165.3	&0.444	&~&	34.5	&0.295	&~&	14.2	&0.194&~	&8.2&0.114\\
&1.25e-5&134.6	&0.575	&~&	16.1	&0.265	&~&	7.8	&0.097&~	&5.1&-0.079\\
\hline
\end{tabular}
\end{table}

\begin{table}[t]
 \renewcommand{\arraystretch}{1.2}
\centering
\caption{Sparsity and Sharpe Ratio of the $\ell_2-\ell_p$-Norm Double-Regularization Model for  International Data with Three Month Estimation Window}
\vskip 0.15cm
\label{tab: 2 and p combined Inter}
\scriptsize
    \begin{tabular}{c c c c c c c c c c c c c}
\hline
~   & ~ &\multicolumn{2}{c}{$\delta=0.1$} & ~ &\multicolumn{2}{c}{$\delta=0.2$} & ~ &\multicolumn{2}{c}{$\delta=0.3$}& ~ &\multicolumn{2}{c}{$\delta=0.4$} \\ \cline{3-4} \cline{6-7}\cline{9-10} \cline{12-13}
~ & $\lambda$ & Spar & SRatio&~ & Spar & SRatio & ~& Spar & SRatio&~& Spar &SRatio \\ \hline
$m_0 =0.00\%$ \quad &   0& 749.8	&0.569&~	&	750.0 &0.528	&	~&	749.8	&0.509&~	&		749.9&0.497	\\
 &  5.0e-7  &265.0	&0.611&~	&	84.8	&0.477&~	&	60.8	&0.443&~	&	31.6&0.439	\\
& 2.0e-6 &109.0	&0.58	&~&	23.4	&0.452&~	&	14.9	&0.439	&~&		8.1&		0.419	\\
& 4.5e-6 &65.7	&0.578&~	&	12.3	&0.461&~	&	5.1	&0.417&~	& 3.8&0.435\\	
&8.0e-6 &60.4	&0.602&~	&	7.67	&0.462&~	&	3.4	&0.415&~	&2.4 &0.433\\	
&1.25e-5&32.2	&0.62	&~&	5.4	&0.407&~	&	2.6	&0.424&~	&	2.25& 0.433\\	
\\
$m_0 =0.04\%$ \quad &0&  749.9	&0.615&~	&	749.9	&0.554&~	&	750.0	&0.524&~	&	750.0&0.508\\
 &  5.0e-7  &270.8	&0.619&~ &	83.8&	0.485	&~&	48.7	&	0.445&~	&32.4&	0.437\\
& 2.0e-6  &115.0   &0.606	&~ &24.5&	0.459	&~&	12.0	&	0.441&~	&8.2&	0.418\\
& 4.5e-6 &67.3   &0.599	&~&	12.58	&0.46 &~&	5.3	&	0.416&~	&3.8&	0.429\\
&8.0e-6 &46.5&0.621	&~&	8.0	&0.456&~&  3.4	&	0.413	&~      &2.4&	0.434	\\
&1.25e-5&33.5&0.628	&~&	5.5	&0.404& ~  &2.6	&	0.423	&~&	2.3&	0.433\\

\\
$m_0 =0.08\%$ \quad &0& 749.7&0.637&~	&	750.0&	0.588&~		&	749.8&	0.544	&~&	749.9&	0.524	\\
 &  5.0e-7  &291.8&0.649&~	&	88.2	&	0.486&~&	49.3	&	0.446&~&	38.0&	0.439\\
& 2.0e-6 &128.6&0.692	&~&	26.3	&	0.467&~&	12.3	&	0.434&~&	8.3&	0.417\\
& 4.5e-6 &87.3&0.631&~&	13.1	&	0.455&~&	5.3&	0.414&~		&	4.17&	0.429\\
&8.0e-6 &52.8&0.623&~	&	8.5	&	0.461&~&	3.5&	0.41	&~&	2.42&	0.434 \\				
&1.25e-5&38.8&	0.625	&~&	5.8&	0.4&~		&	2.6	&	0.4&~&	2.3	&	0.433\\		
\hline
\end{tabular}
\end{table}

\section{Discussions and Conclusions}
\label{sec:discussionsandconclusions}
\setcounter{equation}{0}
\medskip

\subsection{ $\ell_p$-norm regularized Dynamic Portfolios}
A closely related application to our model is the dynamic portfolio selection. Instead of seeking a sparse portfolio, we are looking for a sparse \emph{adjustment} to an already existing portfolio.  Consider the following cardinality constrained optimization model.
\begin{equation}\label{dynamic}
\begin{array}{rl}
\min & \frac{1}{2}x^TQx-c^Tx\\[0.2cm]
\mbox{s.t.}          &                  e^Tx            =1\\[0.1cm]
       &                       x             \ge 0\\[0.1cm]
       &            \|x-a\|_0             \le K,
\end{array}
\end{equation}
Here the $a$-vector is a feasible portfolio ($e^Ta = 1$ and $a \ge 0$), representing the current state of our dynamic portfolio.  Similar to the Markowitz model, the dynamic portfolio has found many applications.  One is the situation where implementing the portfolio takes a significant amount of time (perhaps we must execute our orders sequentially with long delays in-between) and we wish our first orders to constitute an near-optimal portfolio.  Another is the situation where our estimates $Q$ and $c = \phi m$ are themselves varying over time, enough to warrant a re-balancing, yet we still have limits on trading---either due to transaction costs or structural limitations.

This model has a non-differentiable point in the middle of the feasible region ($x=a$), but can be reformulated (by substitution: $y=x-a$) to achieve a model very similar to the non-dynamic sparse portfolio model:
\begin{equation}\label{dynamic2}
\begin{array}{rl}
\min & \frac{1}{2}y^TQy+Qa^Ty-c^Ty\\
\mbox{s.t.}     &                  e^Ty            =0\\
                     &                       y             \ge -a\\
                    &            \|y\|_0               \le K,
\end{array}
\end{equation}
We note that the objective function is still a quadratic function, and that the constraints are also of the same shape.  Instead of solving the original model \eqref{dynamic2}, we consider the following $p$ norm regularized dynamic Markowitz model
\begin{equation}\label{dynamiclp}
\begin{array}{rl}
\min & \frac{1}{2}y^TQy+(a^T Q-c^T)y +\lambda \|y\|_p^p\\[0.3cm]
\mbox{s.t.}          &                  e^Ty            =0,\\[0.2cm]
       &                       y             \ge -a.
\end{array}
\end{equation}
By letting $y=y^+ -y^-$ and using the concavity of $\|\cdot\|_p^p$, we know the regularized model
\eqref{dynamiclp} can be equivalently written as
\begin{equation}\label{dynamiclp1}
\begin{array}{rl}
\min & \frac{1}{2}(y^+-y^-)^TQ(y^+-y^-)+(a^TQ-c^T)(y^+-y^-) +\lambda \|y^+\|_p^p+\lambda\|y^-\|_p^p\\[0.2cm]
\mbox{s.t.}          &        e^Ty^+-e^Ty^-        =0,\\[0.1cm]
                          &          y^+ -y^-      \ge -a,\\[0.1cm]
                         & y^+\geq 0, \,y^-\geq 0,
\end{array}
\end{equation}
which can be further simplified to the following model
\begin{equation}\label{dynamiclp2}
\begin{array}{rl}
\min & \frac{1}{2}(y^+-y^-)^TQ(y^+-y^-)+(a^TQ-c^T)(y^+-y^-) +\lambda \|y^+\|_p^p+\lambda\|y^-\|_p^p\\[0.2cm]
\mbox{s.t.}          &                  e^Ty^+-e^Ty^-        =0,\\[0.1cm]
                          &                    y^+\geq 0, \,\,\,0\leq y^-   \leq a,
\end{array}
\end{equation}
Similar as the non-dynamic $\ell_p$- norm portfolio model, this resulting $\ell_p$-norm model can also be solved by the second order interior interior
point method.

\subsection{Conclusions}

In this paper, we propose an $\ell_p$-norm regularized model with/without shortsale constraints to seek near-optimal sparse portfolios
to  reduce the complexity of portfolio implementation and management. We also study the impact of the $\ell_1$ and $\ell_2$ norms and their cross-effects on overfitting. Theoretical results is established to guarantee the sparsity of the
 novel portfolio strategy. Computational evidence also clearly shows that the $\ell_p$-norm regularized portfolio is able to
choose sparsity with completely flexibility while still maintaining satisfactory out-of-sample performance---comparable to that of the NP cardinality-constrained portfolios.  

We find that the $\ell_1$-norm can be viewed as a prior on the optimal level of portfolio leverage;  a small $\ell_1$-penalty can improve performance. The $\ell_1$ norm greatly reduces the feasible region helping algorithms converge quickly.  It also is shown to be synonymous with leverage---a very important financial term and quantity of great theoretical interest.  

Meanwhile the $\ell_2$-norm can be viewed as a prior on the estimated covariances; we find that a large $\ell_2$-penalty can greatly improve performance, It also could improve tractability by bounding the feasible region.  And $\ell_2$-norm and the $\ell_p$-norm have positive cross-effects on performance---the combined model consistently portfolios outperformed all others.
  
Generally, when we do not pursue the most sparse portfolio,then the cost of sparsity is low---especially when the original portfolio of stocks is diverse.  And our research provides a toolset to evaluate the tradeoffs between sparsity and out-of-sample performance.

Our models also importantly provide a theoretical framework.  In this framework, sparsity can be studied in relation to leverage, correlation, Sharpe-Ratio and financial theory, where both practical bounds and qualitative insights can be made.  
\section{Appendix}
\setcounter{equation}{0}
\medskip
\subsection {Appendix I: Proofs of the Propositions}
\medskip
{\bf  Proof of Theorem \ref{secondordertheorem}.} Since the second-order necessary condition of \eqref{MV-lp} holds at the point $\bar{x}$, the sub-Hessian matrix of the objective function corresponding to the indices $\bar{P}$
\[
\bar{Q}-\frac{\lambda}{4}\bar{X}^{-3/2} \succeq 0
\]
on the null space of $e$. This means the projected Hessian matrix
\[ \left(I-\frac{1}{K}ee^T\right)
\left(\bar{Q}-\frac{\lambda}{4}\bar{X}^{-3/2}\right)
\left(I-\frac{1}{K}ee^T\right)
\]
is positive semidefinite. By direct calculation, we know that the $i$th diagonal entry of the projected Hessian matrix is given by
\begin{equation}\label{diagonalentry}
L_i -\frac{\lambda}{4}\left( (\bar{x}_i)^{-3/2}\left(1-\frac{2}{K}\right)+\frac{\sum_{j\in \bar{P}}(\bar{x}_j)^{-3/2}}{K^2}\right)\ge 0,
\end{equation}
and also the trace of projected Hessian matrix
\[\sum_{i\in \bar{P}}L_i-\frac{\lambda}{4}\frac{K-1}{K}\sum_{i\in \bar{P}}(\bar{x}_i)^{-3/2}\ge 0.\]
The quantity $\sum_{i\in \bar{P}}(\bar{x}_i)^{-3/2}$, with $\sum_{i\in \bar{P}}\bar{x}_i=1$,
achieves its minimum at $\bar{x}_i=1/K$ for all $i\in \bar{P}$ with the minimum value $K\cdot K^{3/2}$. Thus,
\[\frac{\lambda}{4}(K-1)K^{3/2}\le \sum_{i\in \bar{P}}L_i,\]
or
\[(K-1)K^{3/2}\le \frac{4\sum_{i\in \bar{P}}L_i}{\lambda},\]
which complete the proof of the first claim. Moreover, from (\ref{diagonalentry}) we have
\[\frac{\lambda}{4}\left( (\bar{x}_i)^{-3/2}\left(1-\frac{2}{K}\right)+\frac{\sum_{j\in \bar{P}}(\bar{x}_j)^{-3/2}}{K^2}\right)\le L_i.\]
Or
\[\frac{\lambda}{4}\left( (\bar{x}_i)^{-3/2}\left(1-\frac{1}{K}\right)^2+\frac{\sum_{j\in \bar{P},j\ne i}(\bar{x}_j)^{-3/2}}{K^2}\right)\le L_i,\]
which implies
\begin{equation}\label{bound1}
\frac{\lambda}{4}(\bar{x}_i)^{-3/2}\left(1-\frac{1}{K}\right)^2\le L_i.
\end{equation}
Hence, if $L_i=0$, we must have $K=1$ so that $\bar{x}_i$ is the only non-zero entry in $\bar{x}$ and $\bar{x}_i=1$. Otherwise, from \eqref{bound1}, we have the desired second statement in the theorem.

\medskip

\noindent {\bf Proof of Theorem \ref{secondordertheorem1}.}）\,\,(i) Assume the contrary that $\bar{P}^+ \cap \bar{P}^-\neq \emptyset$.
Then there exists an index $j$ such that $\bar{x}_j^+>0$ and $\bar{x}_j^->0$. Let $\lambda_1$ and $\lambda_2\, (\leq 0)$ be the optimal Lagrangian multiplier associated with the constraints of \eqref{MV-lpreg-l1con}. Since $(x^+, x^-)$ is a KKT point of  \eqref{MV-lpreg-l1con}, it holds that
\begin{equation}
\left\{
\begin{array}{c}
\displaystyle \left[Q(\bar{x}^+ -\bar{x}^-)\right]_i - c_i + {\lambda\over 2\sqrt{(\bar{x}^+)_i}} -\lambda_1 -\lambda_2 =0\\
\displaystyle \left[Q(\bar{x}^- -\bar{x}^+)\right]_i + c_i + {\lambda\over 2\sqrt{(\bar{x}^-)_i}} +\lambda_1 -\lambda_2 =0
\end{array}
\right..
\end{equation}
By adding the two equalities above, we have
\begin{equation}\label{contrary}
 {\lambda\over 2\sqrt{(\bar{x}^+)_i}}   + {\lambda\over 2\sqrt{(\bar{x}^-)_i}}  - 2\lambda_2 =0.
\end{equation}
However, since $(\bar{x}^+)_i >0,\,\,  (\bar{x}^-)_i >0$ and $\lambda_2\leq 0$, the equality \eqref{contrary}
cannot hold. This contradiction shows that $\bar{P}^+ \cap \bar{P}^- \neq \emptyset$.
(ii,iii) Since the proof of the remainder parts of this theorem is similar to that of Theorem 1, we omit the details.

\medskip

\noindent {\bf Proof of Theorem \ref{secondordertheorem2} .}）\,\,The proof of this theorem is similar to that of Theorem 1. We omit the details.

\medskip

\subsection{Appendix II: Polynomial Time Interior Point Algorithms}
Most nonlinear optimization solvers can only guarantee to
compute a first-order KKT solution. In this section, we extend the interior-point algorithm
described in \cite{Bian:2012} to solve the following generally $\ell_p$-norm regularized model
\begin{equation}\label{MV-lp-general}
\begin{array}{rl}
\min &\displaystyle  f(x):=\frac{1}{2}\,\,x^TQx - c^Tx + \lambda\|x\|^p_p \\[0.3cm]
\mbox{s.t.}       &                  Ax            =b,\\[0.1cm]
       &                       x             \ge 0,
\end{array}
\end{equation}
where $A$ is a matrix in $\Re^{p\times n}$, $b$ is a vector in $\Re^p$ and the feasible region is strictly feasible. For simplicity, we fix $p= {1\over 2}$.

Naturally, we would start from an interior-point feasible
solution such as the analytical  of the feasible set, and let the iterative
algorithm to decide which entry goes to zero.
This is the basic idea of affine scaling algorithm developed in \cite{Bian:2012} for regularized nonconvex programming. The algorithm starts from an initial interior-point solution, then follows an interior feasible path and finally converges to either a global minimizer or a second-order KKT solution. At each step, it chooses a new interior point which produces a reduction to the objective function by an affine-scaling trust-region iteration.

Specifically, give an  interior point $x^k$ of the feasible region, the algorithm looks for an objective reduction by a update from $x^k$ to $x^{k+1}$. Let $d^k$ be a vector in $\Re^p$ satisfying $Ad^k=0$ and $x^{k+1}:=x^k+d^k>0$.
Using the second Taylor expansion of $f(\cdot)$, we know
\[
f(x^{k+1}) \approx
f(x^k)+\frac{1}{2}(d^k)^T\big(Q-\frac{\lambda}{4}(X^k)^{-3/2}\big) d^k
+\big(Qx^k-c+\frac{\lambda}{2\sqrt{x^k}}\big)^Td^k,
\]
where $X^k = {\rm Diag}(x^k)$. For given $\varepsilon\in (0,1]$, we solve the ellipsoidal trust-region constrained problem
\[
\begin{array}{rl}
\min  &  \displaystyle \frac{1}{2}(d^k)^T\big(Q-\frac{\lambda}{4}(X^k)^{-3/2}\big) d^k
+\big(Qx^k-c+\frac{\lambda}{2\sqrt{x^k}}\big)^Td^k\\[0.3cm]
\mbox{s.t.}   &  Ad^k =0, \\[0.2cm]
                          &   \|X_k^{-1}d^k\|^2\leq \beta^2\varepsilon<1,
\end{array}
\]
to obtain the direction $d^k$. By letting $\tilde{d^k}=X_k^{-1}d^k$, we can recast the above
ellipsoidal trust-region constrained problem above as a ball-constrained quadratic problem
\begin{equation}\label{ipa}
\begin{array}{rl}
\min & \displaystyle  \frac{1}{2}(\tilde{d}^k)^TX^k\big(Q
-\frac{\lambda}{4}(X^k)^{-3/2}\big)X^k\tilde{d}^k+\big(Qx^k-c+
\frac{\lambda}{2\sqrt{x_k}}
\big)^TX^k \tilde{d}^k,\\[0.3cm]
\mbox{s.t.}&   AX_k\tilde{d}^k  =0, \\[0.1cm]
                          &   \|\tilde{d}^k\|^2 \leq \beta^2\varepsilon.
\end{array}
\end{equation}
Note that problem (\ref{ipa}) can be solved efficiently  even when it
is nonconvex (see \cite{Bian:2012}).

Let
$\widetilde{Q}^k = X_kQX_k -\frac{\lambda}{4}\sqrt{X^k}$
and
$\tilde{c}^k = X_k(Qx^k -c) + \frac{\lambda}{2}\sqrt{x^k}$.
If $\widetilde{Q}_k$ is semidefinite, the solution $\tilde{d}^k$ of problem \eqref{ipa} satisfies the following necessary and sufficient conditions:
\begin{equation}\label{KKTdef}
\left\{
\begin{array}{l}
(\widetilde{Q}^k+ \mu_kI)\tilde{d}^k - (AX^k)^Ty_k = -\tilde{c}^k,\\[0.1cm]
A X^k\tilde{d}^k =0,\\[0.1cm]
\mu_k\geq 0,\,\|\tilde{d}^k\|^2\leq \beta^2\varepsilon,\, \mu_k(\|\tilde{d}^k\|^2- \beta^2\varepsilon)=0.
\end{array}\right.
\end{equation}
In the case that $\widetilde{Q}^k$ is indefinite, it holds that
\begin{equation}\label{KKTindef}
\left\{
\begin{array}{l}
(\widetilde{Q}^k+ \mu_kI)\tilde{d}^k - (AX^k)^Ty_k = -\tilde{c}^k,\\[0.1cm]
A X_k\tilde{d}^k =0,\\[0.1cm]
\mu_k\geq 0,\, N_k^T\widetilde{Q}^kN_k +\mu_kI \succeq 0,\\[0.1cm]
\|\tilde{d}^k\|= \beta\sqrt{\varepsilon},
\end{array}\right.
\end{equation}
where $N_k$ is an orthogonal basis spanning the space of $X^kA^T$.

To evaluate the performance of the affine scaling method, we need  the definitions of $\varepsilon$ scaled first-order and second-order KKT solutions. $x^*$ is said to be an $\epsilon$ scaled first-order KKT solution of
\eqref{MV-lp-general} if there exists a $y^*\in \Re^p$ such that
\begin{equation}\label{ekkT1}
\left\{
\begin{array}{l}
\displaystyle \|X^*(Qx^*-c) +\frac{\lambda}{2}\sqrt{x^*}-X^*A^Ty^*\| \leq \epsilon,\\[0.1cm]
A x^* =b,\\[0.1cm]
x^* \geq 0.
\end{array}\right.
\end{equation}
Furthermore, if $\displaystyle X^*QX^*-\frac{\lambda}{4}\sqrt{X^*}+\sqrt{\epsilon} I$ is also semidefinite on the null space of $X ^*A^T$, we call $x^*$ an $\epsilon$ scaled second-order KKT solution. If $\varepsilon =0$, the $\varepsilon$ scaled first-order KKT solution reduces to
\[
X^*(Qx^*-c) +\frac{\lambda}{2}\sqrt{x^*}- X^*A^Ty^* =0,
\]
which is exactly the first-order condition of \eqref{MV-lp-general}.
In this case, the $\varepsilon$ scaled second-order condition
collapses to
\begin{equation}\label{second}
N^T X^*QX^*N  -\frac{\lambda}{4} N^T\sqrt{X^*}N\succeq 0
\end{equation}
where  $N$ is an orthogonal basis spanning the space of $X^*A^T$. By direct computation, we know
\eqref{second} recovers exactly the second-order optimality condition of  problem \eqref{MV-lp-general}.

For the convergence analysis of our proposed interior-point algorithm, we make the following  standard assumption.
For any given $x^0\geq 0$ such that $Ax=b$, there  exists $R\geq 1$ such that
$$
\sup \{\|x\|_\infty: f(x)\leq f(x_0), Ax=b, x\geq 0\}\leq R.
$$

Under the assumption above, we are able to establish the next theorem showing that the affine scaling is able to obtain either an $\varepsilon$-scaled second-order KKT solution or an $\varepsilon$ global minimizer in polynomial time.

\begin{theorem}
Let $\varepsilon\in (0,1]$. There exists a positive number $\tau$ such that the proposed second-order interior point obtains either an $\varepsilon$ scaled second-order KKT solution or $\varepsilon$ global minimizer of \eqref{MV-lp-general} in no more than $O(\varepsilon^{-3/2})$ iterations provided that $\beta\in (0, \,\tau)$.
\end{theorem}
\begin{proof} 
With loss of generality, we assume the radius $R=1$ in the assumption. To proceed the proof  of this theorem, we first introduce the following Lemma.
\begin{lemma}
If $\mu_k > \lambda/6 \|\tilde{d}^k\|$ holds for all $k=0,1,2,\ldots$, then the second-order
interior point algorithm produces an $\varepsilon$ global minimizer of  \eqref{MV-lp-general} in at most $O(\varepsilon^{-3/2})$ iterations.
\end{lemma}
\begin{proof}
By the Taylor expansion of $\sqrt{\cdot}$, it is easily to show that
\[
f(x^{k+1}) - f(x^k) \leq \frac{1}{2}
\big\langle \tilde{d}^k,\, \widetilde{Q}^k \tilde{d}^k\big\rangle  +\big\langle \tilde{c}^k,\,\tilde{d}^k\rangle +\frac{3\lambda}{48} \|\tilde{d}_k\|^3.
\]
From \eqref{KKTdef} and \eqref{KKTindef}, then
\begin{equation}
\label{eq:decrease}
\begin{array}{rl}
f(x^{k+1}) - f(x^k) &\displaystyle \leq \frac{1}{2}
\big\langle \tilde{d}^k,\, \widetilde{Q}^k \tilde{d}^k\big\rangle  +\big\langle -\widetilde{Q}^k\tilde{d}^k-\mu_k \tilde{d}^k+(AX_k)^Ty_k,\,\tilde{d}^k\rangle +\frac{3\lambda}{48} \|\tilde{d}^k\|^3\\[0.2cm]
&\displaystyle  =-\frac{1}{2}\tilde{d}^k \widetilde{Q}^k \tilde{d}^k -\mu_k \|\tilde{d}^k\|^2+\frac{3\lambda}{48} \|\tilde{d}^k\|^3\\[0.2cm]
&\displaystyle= -\frac{1}{2} (v^k)^T(N_k)^T\widetilde{Q}^kN_k v^k-\mu_k \|\tilde{d}^k\|^2 +\frac{3\lambda}{48} \|\tilde{d}^k\|^3\\[0.2cm]
& \displaystyle \leq
\frac{\mu_k}{2} \|v^k\|^2-\mu_k \|\tilde{d}^k\|^2 +\frac{\lambda}{48} \|\tilde{d}^k\|^3
\\[0.2cm]
&\displaystyle =  -\frac{\mu_k}{2}\|\tilde{d}^k\|^2 +\frac{3\lambda}{48} \|\tilde{d}^k\|^3\\[0.2cm]
&\displaystyle \leq -\frac{1}{8} \mu_k\|\tilde{d}_k\|^2,
\end{array}
\end{equation}
where the second inequality follows from the semidefiniteness of $(N_k)^T(\widetilde{Q}^k)N_k + \mu_k I$
and the last inequality comes from the relationship that $\|\tilde{d}_k\|< 6\mu_k/\lambda$.
Combining \eqref{eq:decrease} with the fact that $\|\tilde{d}^k\| =\beta\sqrt{\varepsilon}$
due to $\mu_k>0$, we further have
\[
f(x^k) - f(x^0)\leq -\frac{1}{8} \sum_{j=0}^{k-1} \mu_j\|\tilde{d}_j\|^2\leq -\frac{\lambda}{48}k\big({\beta^2\varepsilon}\big)^{3/2}
\]
and hence the interior-point algorithm produces an $\varepsilon$ global minimizer in $O(\varepsilon^{-\frac{3}{2}})$ iterations. 
\end{proof}

In what follows, we pay more attentions to the case where
 $\mu_k\leq \lambda/6\|\tilde{d}^k\|$ for some $k$.
\begin{lemma}
Let $ \beta\leq \min\{\frac{1}{2},\sqrt{2\over \lambda},\frac{3}{(18\sqrt{2}+2)\lambda}\}$. If there exists some $k$ such that $\mu_k\leq \frac{\lambda}{6}\|\tilde{d}^k\|$, then $x^{k+1}$ is an $\varepsilon$ second-order
KKT solution of  \eqref{MV-lp-general}.
\end{lemma}
\begin{proof} (i) We firstly show $x^{k+1}$ is an $\varepsilon$ scaled first order KKT solution
when $\beta$ is restricted into the special range.
From \eqref{KKTdef} and \eqref{KKTindef}, it follows that
\[
-\mu_k \tilde{d}^k = X^k({Q} x^{k+1}-c) -\frac{\lambda}{4}\sqrt{X^k}\tilde{d}^k + \frac{\lambda \sqrt{x^k}}{2} -X^kA^Ty^k,
\]
which implies that
\[
Qx^{k+1} -c -A^Ty^k = \frac{\lambda}{4}(X^k)^{-1/2}\tilde{d}^k
-\frac{\lambda}{2} (x^k)^{-1/2}
-\mu_k (X^k)^{-1}\tilde{d}^k.
\]
Therefore, we have
\begin{equation}
\begin{array}{rl}
&\displaystyle \|X^{k+1} (Q x^{k+1} - c)  +\frac{\lambda\sqrt{x^{k+1}}}{2} - X^{k+1}A^Ty^k\|\\[0.1cm]
=&\displaystyle \|\frac{\lambda\sqrt{x^{k+1}}}{2} -  \frac{\lambda}{2}X^{k+1}(x^{k})^{(-1/2)}
+\frac{\lambda}{4}X^{k+1}({X^k})^{-1/2}\tilde{d}^k-\mu_k X^{k+1}(X^k)^{-1}\tilde{d}^k\|\\[0.25cm]
\leq &\displaystyle \|\frac{\lambda\sqrt{x^{k+1}}}{2} -\frac{\lambda}{2}X^{k+1}(x^{k})^{(-1/2)}
+\frac{\lambda}{4}X^{k+1}({X^k})^{-1/2}\tilde{d}^k\|+\mu_k\| X^{k+1}(X^k)^{-1}\tilde{d}^k\|\\[0.25cm]
\leq & \displaystyle \frac{\lambda}{2} \|\sqrt{X^k}\|_\infty\|\sqrt{\tilde{d}^k +e} - e -\frac{1}{2}\tilde{d}^k + \frac{1}{2}(\tilde{d}^k)^2\|+ \mu_k \|\tilde{d}^k\|(1+ \|\tilde{d}^k\|)
\end{array}
\end{equation}
Since the condition $\mu_j > {\lambda\over 6}\|\tilde{d}^j\|$ holds for $j=0,1,2\ldots, k-1$, by the proof of Lemma 1,
we have $f(x^k) \leq f(x^0)$, which together with Assumption 1 implies  $\|x^k\|_\infty \leq 1$. Moreover, we know from
the proof of Lemma 4 in \cite{Bian:2012} that
\[
\|\sqrt{\tilde{d}^k+e}-e -\frac{1}{2}\tilde{d}^k + \frac{1}{2}({\tilde d}^k)^2\|\leq  \frac{1}{2} \|\tilde{d}^k\|^2
\]
and hence
\[
\begin{array}{rl}
      &\displaystyle \|(X^{k+1}) (Q x^{k+1} - c)  +\frac{\lambda\sqrt{x^{k+1}}}{ 2} - X^{k+1}A^T y^k\|\\[0.3cm]
\leq&\displaystyle  \frac{\lambda}{4}\|\tilde{d}^k\|^2 + \frac{3}{2}\mu_k \|\tilde{d}^k\|\leq \frac{\lambda}{2}\|\tilde{d}^k\|^2\leq \varepsilon,
\end{array}
\]
which means $x^{k+1}$ is an $\varepsilon$ scaled first-order KKT solution.

(ii) Again from \eqref{KKTdef} and \eqref{KKTindef}, we know that
\[
X^kQX^k -\frac{\lambda}{4}\sqrt{X^k}+\mu_k I
\]
is positive semidefinite on the null space that  $X^kA^T$.  Let $N_k$ be the orthogonal basis of this null space and it therefore holds
\begin{equation}\label{secondKKT1}
N_k^T(X^kQX^k -\frac{\lambda}{4}\sqrt{X^k})N_k\succeq -\mu_kI \succeq -\frac{\lambda}{6}\beta \sqrt{\varepsilon}I.
\end{equation}
Clearly,  $N_{k+1}:= (X^{k+1})^{-1}X^k N_k$ is a basis of the null space of $X^{k+1}A^T$. By simple algebraic computation,
we can easily obtain that
\begin{equation}\label{secondKKT3}
\begin{array}{rl}
&\displaystyle N_{k+1}^T\left[X^{k+1}QX^{k+1} - \frac{\lambda}{4}\sqrt{X^{k+1}} + \sqrt{\varepsilon}I\right]N_{k+1}\\[0.4cm]
=&\displaystyle  N_k^T(X^kQ X^k-{\lambda\over 4}\sqrt{X^k})N_k + \sqrt{\varepsilon} N_k^T\left[X_{k+1}^{-2}(X^k)^2\right] N_k\\[0.2cm]
&\displaystyle +{\lambda\over 4} N_k^T \sqrt{X^k} \big[I- (X^k)^{3/2} (X^{k+1})^{-3/2}\big]N_k \\[0.4cm]
\displaystyle \succeq &   -\frac{\lambda}{6}\beta \sqrt{\varepsilon}I
+ \sqrt{\varepsilon}N_k^T (I+ D_k)^{-2}N_K + {\lambda\over 4}N_k^T \sqrt{X^k}\big[I - (I+D_k)^{-3/2}\big]N_k
\end{array}
\end{equation}
where $D_k = {\rm Diag}(\tilde{d}^k)$. Since $\|\tilde{d}^k\| \leq \beta \sqrt{\varepsilon}\leq {1\over 2}<1$, we know
\begin{equation}\label{eq:d1}
(I+ D_k)^{-2}\succeq (1+ \beta \sqrt{\varepsilon})^{-2}I \, \succeq {1\over 4}I
\end{equation}
and
\begin{equation}\label{eq:d2}
I - (I + D_k)^{-3/2}\succeq \big[1- (1- \beta \sqrt{\varepsilon})^{-3/2}\big] I.
\end{equation}
Moreover,  the mean-value theorem applied to the function $x^{-3/2}$  yields that
\[
1- (1- \beta \sqrt{\varepsilon})^{-3/2}  = -{3\over 2} \beta \sqrt{\varepsilon} \theta^{-5/2},
\]
where $\theta$ is in the open interval $(1-\beta\sqrt{\varepsilon}, \,1 )$. Note that $\beta\sqrt{\varepsilon} \leq{1\over 2}$,
then it holds that
\begin{equation}\label{eq:d3}
1- (1- \beta \sqrt{\varepsilon})^{-3/2} \geq -6\sqrt{2}\beta \sqrt{\varepsilon}.
\end{equation}
By substituting \eqref{eq:d1}, \eqref{eq:d2} and \eqref{eq:d3} into \eqref{secondKKT3}, we immediately get that
\[
N_{k+1}^T\left[X^{k+1}QX^{k+1} - \frac{\lambda}{4}\sqrt{X^{k+1}} + \sqrt{\varepsilon}I\right]N_{k+1}
\succeq ({1\over 4}-{3\sqrt{2}\beta \lambda \over 2 }-{ \lambda \over 6}\beta)\sqrt{\varepsilon} I \succeq 0.
\]
Thus $x^{k+1}$ is an $\varepsilon$ scaled second-order KKT solution.
\end{proof}
\paragraph{}
According to the above two lemmas, we know the proposed second order interior point obtains either an $\varepsilon$ scaled second KKT solution or $\varepsilon$ global minimizer in no more than $O(\varepsilon^{-3/2})$ iterations provided that  $\beta_k\leq \min\{\frac{1}{2},\sqrt{2\over \lambda},\frac{3}{(18\sqrt{2}+2)\lambda}\}$. This completes the proof
of this Theorem.
\end{proof}

\end{document}